\documentclass[reprint,onecolumn,notitlepage]{revtex4-1}
\usepackage[english]{babel}
\usepackage{amssymb, amsthm, amsmath}
\usepackage{hyperref}
\usepackage{braket}
\usepackage{graphicx}
\usepackage[margin=1in]{geometry}
\usepackage{bbm}
\usepackage{algorithm}
\usepackage{color}
\floatname{algorithm}{Protocol}

\def\EQ#1{\begin{eqnarray}#1\end{eqnarray}}
\def\H {{\mathcal H}}
\newcommand{\AR}[2][c]{$$\begin{array}[#1]{lllllllllllllll}#2\end{array}$$}
\newtheorem{Theorem}{Theorem}
\newtheorem{Lemma}{Lemma}

\newtheorem{definition}{Definition}
\DeclareMathOperator{\Tr}{\operatorname{Tr}}

\begin{document}
\title{Garbled Quantum Computation}

\author{Elham Kashefi}
\affiliation{School of Informatics, University of Edinburgh,10 Crichton Street, Edinburgh EH8 9AB, UK}
\affiliation{CNRS LTCI, Departement Informatique et Reseaux, UPMC - LIP6, 4 place Jussieu 75252 PARIS CEDEX 05, France}
\author{Petros Wallden}
\affiliation{School of Informatics, University of Edinburgh,10 Crichton Street, Edinburgh EH8 9AB, UK}

\begin{abstract}
The universal blind quantum computation protocol (UBQC) \cite{bfk} enables an almost classical client to delegate a quantum computation to an untrusted quantum server (in form of a garbled quantum computation) while the security for the client is unconditional. In this contribution we explore the possibility of extending the verifiable UBQC \cite{fk,KW15}, to achieve further functionalities as was done for classical garbled computation. First, exploring the asymmetric nature of UBQC (client preparing only single qubits, while the server runs the entire quantum computation), we present a  ``Yao'' type protocol for secure two party quantum computation. Similar to the classical setting \cite{Yao86} our quantum Yao protocol is secure against a specious (quantum honest-but-curious) garbler, but in our case, against a (fully) malicious evaluator. Unlike the protocol in  \cite{DNS10}, we do not require any online-quantum communication between the garbler and the evaluator and thus no extra cryptographic primitive.  This feature will allow us to construct a simple universal one-time compiler for any quantum computation using one-time memory, in a similar way with the classical work of \cite{GKR08} while more efficiently than the previous work of \cite{BGS13}.
\end{abstract}


\maketitle

\section{Introduction}

Future information and communication networks will 
consist of both classical and quantum devices, some of which are expected to be dishonest. 
These devices will have various different 
functionalities, ranging from simple routers to servers executing quantum algorithms. Within the last few years, anticipating this development 
has led to the formation and growth
of the field of delegated quantum computing \cite{Childs,AS06,bfk,abe,ruv2,BGS13}. Among them is the universal blind quantum computation (UBQC) protocol of ~\cite{bfk} which is developed based on 
the measurement-based quantum computation model (MBQC) \cite{onewaycomputer} that appears 
as the most promising physical implementation 
for a networked architecture \cite{NFB14}. In the UBQC framework, the only quantum requirement for the client 
is the offline creation of random single qubit states, which is a currently available technology, and has been 
demonstrated experimentally \cite{BKBFZW11}.

The MBQC model of computation, can be viewed 
as a set of classical instructions steering a quantum computation performed on a highly entangled quantum state. The classical outcomes of the single-system 
measurements that occur \emph{during} the computation, 
are in general 
randomly distributed bits with no significance 
for the final output of the 
computation. This enables 
one to use relatively basic obfuscation techniques in order to prevent an untrusted operator, 
implementing 
an MBQC computation, to access the true 
flow of information. This key observation has led to 
an entirely new approach to quantum verification  that exploits cryptographic techniques 
\cite{fk,KKD14,gkw2015,kdk2015,abe,DFPR13,ruv2,Mckague13,B15,GWK2017}. The core idea is to encode simple \emph{trap} computations within a target computation, 
run on a remote device. This is done in such a way that the computation is not affected, while in the same time, reveals 
no information to the server. 
The correctness of the overall 
computation is 
tested by verifying that 
the trap computations were done correctly. The latter, being significantly less costly, leads to efficient verification schemes. This approach of quantum verification has been 
recently used to obtain 
specific cryptographic primitives such as  quantum one-time program \cite{BGS13} and secure two-party quantum computation \cite{DNS12}, that are also the main focus 
of this paper. 


\subsection{Our contribution}

We will explore two 
extensions of the verifiable universal blind quantum computing (VUBQC) protocols, that are build based on 
measurement-based quantum computation, 
in order to achieve new functionalities to be implemented in 
the quantum network setting \cite{NFB14}. The essential novelty of our approach is that 
the 
client-server setting 
allows different participants to have 
different quantum technological requirements. 
As a result, in our proposed two-party primitives, one party (garbler or sender) remains as classical as possible (with no need for quantum memory), while it is only the other party (evaluator or receiver) that requires 
access to a quantum computer. Moreover, the required offline initial quantum communication between the participants is 
also limited to exchange of simple single quantum states that can be generated, for example, in any quantum key distribution network (as it was proven recently in \cite{DK16}). Finally, all the utilised  sub-protocols in our schemes could be reduced to the core protocols of offline preparation \cite{DK16} and verifiable universal blind quantum computation \cite{fk,DFPR13} that are both proven to be composably secure in the abstract cryptography framework of \cite{MR11}. 
For simplicity, we only present the stand-alone security proof for our protocols, 
as we follow the framework of \cite{DNS10} and use simulation-based techniques. 
Concretely we present two new protocols:

\begin{enumerate}

\item In Section \ref{sec:2PQC} we present a protocol 
for secure two party quantum computation, that we refer to as QYao. This protocol involves two parties with asymmetric roles, that wish to securely compute any given unitary on a joint quantum input. 
Similar to the classical protocol 
of \cite{Yao86}, in QYao an honest-but-curious (formally defined for the quantum setting in \cite{DNS10}) client, capable of preparing only random single qubit and performing constant depth classical computation, ``garbles'' the entire unitary. The fully malicious server/evaluator, 
receives 
instructions from garbler and 
insert their 
input either using 
a (classical) oblivious transfer (OT), or by 
a quantum input insertion scheme 
secure against the 
honest-but-curious garbler. The evaluator 
performs the computation, 
extracts their own 
output qubits and returns the remaining output qubits to the garbler. 
After the garbler verifies 
the computation, the garbler releases the encryption keys of the evaluator output qubits.  Unlike the classical setting, our proposed QYao protocol is interactive, but importantly, 
only 
\emph{classical} online communication is required. 
One gain we obtain, is a boost in security since, apart from the initial input exchange where classical OT is needed, the rest of the protocol is unconditionally secure. 
If one could 
replace the classical OT's in our QYao protocol with new primitives, such as the unconditionally secure relativistic primitives  
\cite{Kent99}, the security would be extended to fully unconditional. 
In Section \ref{sec:simulators} we prove the security of the protocol in the ideal-real word paradigm using a simulation-based technique. 

\item In Section \ref{sec:non-interactive} we follow the classical approach of \cite{GKR08} and make our QYao protocol non-interactive by using the classical hardware primitive ``one-time memory'' (OTM), which is essentially a non-interactive oblivious transfer. By using OTM the need for
initial OT calls, in our QYao protocol, is also removed leading to a one-time universal compiler for any given quantum computation. Such one-time quantum programs can be executed only once, where the input can be chosen at any time. The one-time programs have a wide range of applications ranging from program obfuscation, software protection to temporary transfer of cryptographic ability \cite{GKR08}. Classically the challenge in lifting the Yao protocol for two-party computation to one-time program, is addressing the issue of malicious adversary. However, our QYao protocol is already secure against a malicious evaluator without any extra primitive or added overhead. Our quantum one-time program is therefore a straightforward translation of QYao. The same technique is applicable to essentially any cryptographic protocols developed in the measurement-based model that require offline quantum preparation and online classical communications. In any such MBQC computation there are exponentially many branches of computation (this is due to the non-deterministic nature of the intermediate single-qubits measurement that occur during 
the computation). However, we prove that a single OTM (of constant size) per computation qubit suffices to make the QYao non-interactive. This is due to the fact that in any constant degree graph-state (the underlying entanglement resources for an MBQC computation), the flow of information acts locally. We prove that this is the case for the 
recourse in the VUBQC protocol introduced in \cite{KW15}, that 
our protocol is based. Hence each classical message between parties, in the interactive case, depends 
on constant number of previous messages which 
allows us to remove this interaction using a simple constant-size OTM. 

\end{enumerate}

\subsection{Related works}

Deriving new quantum cryptographic protocols for functionalities 
beyond the secure key exchange \cite{BBE92} is an active field of research (see a recent review for a summary \cite{BS16}). In particular, our contributions are directly linked to the works on verifiable blind quantum computing (VUBQC) \cite{fk,abe,ruv2}, secure two party quantum computing \cite{DNS10,DNS12} and quantum one-time program and quantum obfuscation \cite{BGS13,GF16}. The main focus in VUBQC research is the ability of a limited verifier to certify the correctness of the quantum computation performed by a remote server. Recently, many such protocols have been developed achieving different quantum technological requirement for the verifier or the prover \cite{KKD14,kdk2015,Mckague13,Broadbent15,HT15}. We have used 
the 
optimal (from the verifier point of view) VUBQC that has the additional property of a local 
and independent trap construction \cite{KW15}. This property allows us to construct a simple yet generic scheme for the server (receiver, evaluator) input insertion and output extraction that could be applicable to other VUBQC schemes and that might lead to further 
properties not present 
in the scheme of \cite{KW15}. 

Due to the impossibility results of \cite{BGS13,GF16}, having extra primitives such as OTM is unavoidable in order to achieve program hiding. However, using 
the MBQC framework with 
the 
verification protocol 
of \cite{KW15} leads to a simpler procedure for 
removing the classical interaction compared to the initial work that pioneered this approach 
\cite{BGS13}. Furthermore, due to the direct utilisation of OTM's instead of classical one-time programs, our scheme could be 
applicable to other VUBQC protocols that might emerge in the future. Another way making protocols non-interactive is based on the security of other classical primitives 
as it was done in \cite{DSS16} using fully homomorphic encryption. 
However, to apply this method for non-interactive secure two-party quantum computation, is important to enable verifiability from the side of the sender, something naturally present in our scheme but not necessarily in other schemes.
Finally the VUBQC framework which naturally separates the classical and quantum part of the computation allows us to construct a client-server scheme for 
secure two-party quantum computation 
that unlike the work of \cite{DNS10} removes the requirement of \emph{any} extra cryptographic primitive that was 
an open question in \cite{DNS10}. Here, we need to clarify that, 
while in the classical Yao with an honest-but-curious garbler 
the primitive of OT is required for the evaluator's input insertion, 
in our QYao this is not the case. This is simply due to the fact that the notion of quantum specious adversary that was formalised in \cite{DNS10} is more restrictive than classical honest-but-curious in certain cases (see Appendix \ref{app:specious}). Hence instead of utilising a classical OT, one could simply devise a quantum communication scheme, as we present here, 
in order to achieve the same goal of secure input insertion 
against a (less powerful) adversarial garbler.  

Finally, the server - client setting can also be exploited to achieve simpler (implementation-wise) multiparty blind quantum computation protocols, where multiple clients use a single server to jointly and securely perform a quantum computation \cite{KP16}.

\section{Preliminaries}

\subsection{Verifiable blind quantum computation}

We will assume that the reader is familiar with the measurement-based quantum computation (MBQC) model \cite{onewaycomputer} that is known to be the same as any gate teleportation model \cite{childs2005unified,mbqc}. In this section we introduce MBQC and use it to revise a blind quantum computation (server performs computation without learning input/output or computation) \cite{bfk} and a verifiable blind quantum computation (client can also verify that the computation was performed correctly) \cite{fk} protocols. The general idea behind the MBQC model is: start with a large and highly entangled generic multiparty state (the resource state) and then perform the computation by carrying out single-qubit measurements. To perform a desired quantum computation, each qubit should be measured in a suitable basis, and this basis is (partly) determined by some default measurement angles $\phi_i$. There is an 
order  that the measurements should occur which is determined by the flow of the computation (see also later), and 
the basis that each qubit is measured
generally depends on the outcomes of previous measurements  (and the default measurement angle $\phi_i$). The resource states used are known as \emph{graph states} as they can be fully determined by a given graph (see details in \cite{hein2004multiparty}). A way  
to construct a graph state given the graph
description is to assign to each vertex of the graph a qubit initially prepared in the state $\ket{+}$ and for each edge of the graph to perform a $\mathrm{controlled-}Z$ gate to the two adjacent vertices. For completeness, in Appendix \ref{app:mbqc} we give the expression for the measurement angles of each qubit, and an example of measurement pattern (graph state and default measurement angles $\phi_i$) that implement each gate from a universal set of gates.

If one starts with a graph state where qubits are prepared in a rotated basis $\ket{+_\theta}=1/\sqrt{2}(\ket{0}+e^{i\theta}\ket{1})$
instead, then it is possible to perform the same computation with the non-rotated graph state by preforming measurements in a similarly
rotated basis. This observation led to the formulation of the 
\emph{universal blind quantum computation} (UBQC) protocol \cite{bfk} which hides the computation in a client-server setting.
Here a client prepares rotated qubits, where the rotation is only known to them. The client sends the qubits to the server, as soon as they
are prepared (hence there is no need for any quantum memory). Finally, the client instructs the server to perform entangling operations according to the graph and to carry out single qubits measurements in suitable angles in order to complete the desired computation (where an extra randomisation $r_i$ of the outcome of the measurements is added). During the protocol the client receives  the 
outcomes of previous measurements and can  
classically evaluate the next measurement angles. Due to the unknown rotation and the extra outcome randomisation, the server does not learn what computation they actually perform.

The UBQC protocol can be uplifted to a verification protocol where the client 
can detect a cheating
server. 
To do so, the client for certain vertices (called dummies) sends states from the set
$\{\ket{0},\ket{1}\}$ which has the same effect as a $Z$-basis measurement on that vertex. In any graph state if a vertex is measured
in the $Z$-basis it results in a new graph where that vertex and all its adjacent edges are removed. During the protocol the server does not
know for a particular vertex if the client sent a dummy qubit or not. This enables the client to isolate some qubits (disentangled
from the rest of the graph). Those qubits have fixed deterministic outcomes if the server followed honestly the
instructions. The positions of those isolated qubits are unknown to the server and the client uses them as traps to test that the server
performs the quantum operations that is requested. This technique led to 
the first universal VUBQC protocol \cite{fk} and was subsequently extended to many other protocols depending on what 
optimised construction was used, what was the required quantum technology and which was the desired level of security. It is clear that the more (independent) traps we have within a graph, the higher the probability of detecting a deviation. In \cite{KW15} the authors gave a construction that, while maintaining an overhead which is linear in the size of the input, introduced multiple traps and in particular of the same number as the computation qubits. We will be using that construction, not only for efficiency reasons, but because this construction is ``local'', and revealing to the server partial information about the graph does not compromise the security. This is important for our QYao protocol, since in this case, the server has input and output, and needs to know in which parts of the graph their input/output belongs. The construction is summarised below:

\begin{enumerate}
\item We are given a base-graph $G$ that has vertices $v\in V(G)$ and edges $e\in E(G)$.
\item For each vertex $v_i$, we define a set of three new vertices $P_{v_i}=\{p^{v_i}_1,p^{v_i}_3,p^{v_i}_3\}$. These are called \emph{primary} vertices.
\item Corresponding to each edge $e(v_i,v_j)\in E(G)$ of the base-graph that connects the base vertices $v_i$ and $v_j$, we introduce a set of nine edges $E_{e(v_i,v_j)}$ that connect each of the vertices in the set $P_{v_i}$ with each of the vertices in the set $P_{v_j}$.


\item We replace every edge in the resulted graph with a new vertex connected to the two vertices originally joined by that edge. The new vertices added in this step are called \emph{added} vertices. This is the \emph{dotted triple-graph} $DT(G)$.
\end{enumerate}
We can see (Figure \ref{figure1}) that each vertex in the $DT(G)$, corresponds to either a vertex (for primary vertices) or an edge (for added vertices) of the base-graph. The precise edge/vertex of the base-graph that each vertex $v\in DT(G)$ belongs is called \emph{base-location}. The nice property of this graph, is that one can reduce this graph to three copies of the base-graph by ``breaking'' some edges (which can be done using dummy qubits). Moreover, the choice 
of which vertex belongs to each of the three graphs, is essentially independent (for different base-locations corresponding to vertices of the base-graph). With this construction, we can eventually use one copy of the base-graph for the computation, while make multiple single traps from the two remaining copies (see Figure \ref{figure1}). The choice of which vertex belongs to which graph is called \emph{trap-colouring}. This is a free choice made by the client, and the fact that the server is ignorant of the actual trap-colouring guarantees the security (see details in \cite{KW15}). 

\begin{definition}[Trap-Colouring]\label{trap colouring} We define trap-colouring to be an assignment of one colour to each of the vertices of the dotted triple-graph that is consistent with the following conditions. 
\begin{enumerate}
\item[(i)] Primary vertices are coloured in one of the three colours white or black (for traps), and green (for computation). There is an exception for input base-locations (see step (v)). 
\item[(ii)] Added vertices are coloured in one of the four colours white, black, green or red. 
\item[(iii)] In each primary set $P_v$ there is exactly one vertex of each colour. 
\item[(iv)] Colouring the primary vertices fixes the colours of the added vertices: Added vertices that connect primary vertices of different colour are red. Added vertices that connect primary vertices of the same colour get that colour.
\item[(v)] For input base-locations, instead of green we have a blue vertex (but all other rules, including how is connected with the other vertices apply in the same way as if it was green).
\end{enumerate}
\end{definition}

\begin{figure}[h]
\includegraphics[width=1\columnwidth]{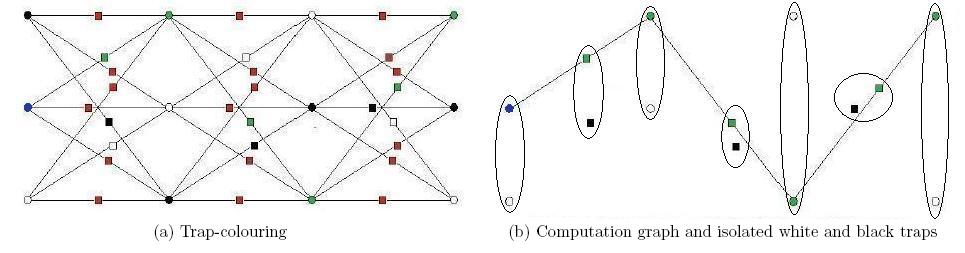}

\caption{Dotted-triple-graph. Circles: primary vertices with base-location of vertex of the base-graph; Squares: added vertices with base-location of edges of the base-graph. (a) Trap-colouring. Blue: input qubits; Green: gate qubits; White/black: trap qubits; Red: wiring qubits. Client chooses the colours randomly for each vertex with base-location of vertex of the base-graph and prepares each qubit individually before sending them one by one to the server to entangle them according to the generic construction. (b) After entangling, the breaking operation defined by the wiring qubits will reduce the graph in (a) to the computation graph and for each vertex a corresponding trap/tag qubits.}
\label{figure1}
\end{figure}

For completeness we give the basic verification protocol from \cite{KW15} that we use. 

\begin{algorithm}[H]
\caption{Verifiable Universal Blind Quantum Computation using dotted triple-graph (with Fault-tolerant Encoding) - Taken from \cite{KW15}}
 \label{prot:KW15}
\noindent A standard labelling of vertices of the dotted triple-graph $DT(G)$ is known to both client and server.\\ \noindent The number of qubits is at most $3N(3c+1)$ where $c$ is the maximum degree of the base graph $G$.\\
\noindent$\bullet$ \textbf{Client's resources} \\
-- Client is given a base graph $G$. The corresponding dotted graph state $\ket{D(G)}$ is generated by graph $D(G)$ that is obtained from $G$ by replacing every edge with a new vertex connected to the two vertices originally joined by that edge.\\
-- Client is given an MBQC measurement pattern $\mathbb{M}_{\textrm{Comp}}$ which: Applied on the dotted graph state $\ket{D(G)}$ performs the desired computation, in a fault-tolerant way, that can detect or correct errors fewer than $\delta/2$.\\
-- Client generates the dotted triple-graph $DT(G)$, and selects a trap-colouring according to definition \ref{trap colouring} which is done by choosing independently the colours for each set $P_v$.\\
-- Client for all red vertices will send dummy qubits and thus performs break operation.\\
-- Client chooses the green graph to perform the computation.\\ 
-- Client for the white graph sends dummy qubits for all added qubits $a^{e}_w$ and thus generates white isolated qubits 
at each primary vertex set $P_{v}$. Similarly for the black graph the client sends dummy qubits for the primary qubits $p^v_b$ and thus generates black isolated qubits 
at each added vertex set $A_{e}$.\\
\noindent -- The dummy qubits position set $D$ is chosen as defined above (fixed by the trap-colouring).\\
\noindent -- A binary string $\mathbf{s}$ of length at most $3N(3c+1)$ represents the measurement outcomes. It is initially set to all zero's.\\
\noindent -- A sequence of measurement angles, $\phi=(\phi_i)_{1\leq i \leq 3N(3c+1)}$ with $\phi_i \in A=\{0,\pi/4,\cdots,7\pi/4\}$, consistent with $\mathbb{M}_{\textrm{Comp}}$. We define $\phi_i'(\phi_i,\mathbf{s})$ to be the measurement angle in MBQC, when corrections due to previous measurement outcomes $\mathbf{s}$ are taken into account (the function depends on the specific base-graph and its flow, see e.g. \cite{bfk}). We also set $\phi'_i = 0$ for all the trap and dummy qubits. 
The Client chooses a measurement order on the dotted base-graph $D(G)$ that is consistent with the flow of the computation (this is known to Server). The measurements within each set $P_v,A_e$ of $DT(G)$ graph are ordered randomly.
\\
\noindent -- $3N (3c+1)$ random variables $\theta_i$ with value taken uniformly at random from $A$.\\
\noindent -- $3N (3c+1)$ random variables $r_i$ and $|D|$ random variable $d_i$ with values taken uniformly at random from $\{0,1\}$. \\
\noindent -- A fixed function $C(i, \phi_i, \theta_i, r_i, \mathbf{s})= \phi_i'(\phi_i,\mathbf{s})+\theta_i+\pi r_i$ that for each non-output qubit $i$ computes the angle of the measurement of qubit $i$ to be sent to the Server.\\
-- Continues
\end{algorithm}

\begin{algorithm}[H]
\caption{continuing: VUBQC with DT(G)}
\noindent$\bullet$ \textbf{Initial Step} \\
-- \textbf{Client's move:} Client sets all the value in $\mathbf{s}$ to be $0$ and prepares the input qubits as
\AR{
\ket e = X^{x_1} Z(\theta_1) \otimes \ldots \otimes  X^{x_l} Z(\theta_l) \ket I
}
and the remaining qubits in the following form
\AR{
\forall i\in D &\;\;\;& \ket {d_i} \\ 
\forall i \not \in D &\;\;\;& \prod_{j\in N_G(i) \cap D} Z^{d_j}\ket {+_{\theta_i}}
}
and sends the Server all the $3N (3c+1)$ qubits in the order of the labelling of the vertices of the graph.\\
-- \textbf{Server's move:} Server receives $3N(3c+1)$ single qubits and entangles them according to $DT(G)$.\\

\noindent$\bullet$ \textbf{Step $i: \; 1 \leq i \leq 3N (3c+1)$}

-- \textbf{Client's move:} Client computes the angle $\delta_i=C(i, \phi_i, \theta_i, r_i, \mathbf{s})$ and sends it to the Server.\\ 
-- \textbf{Server's move:} Server measures qubit $i$ with angle $\delta_i$ and sends the Client the result $b_i$. \\
-- \textbf{Client's move:} Client sets the value of $s_i$ in $\mathbf{s}$ to be $b_i+r_i$. \\

\noindent$\bullet$ \textbf{Final Step:}

-- \textbf{Server's move:} Server returns the last layer of qubits (output layer) to the Client.\\

\noindent$\bullet$  \label{step:Alice-prep} \textbf{Verification} \\
-- After obtaining the output qubits from the Server, the Client measures the output trap qubits with angle $\delta_t=\theta_t+r_t\pi$ to obtain $b_t$.

-- Client accepts if $b_i = r_i$ for all the white (primary) and black (added) trap qubits $i$.

\end{algorithm}



\subsection{Two-party quantum protocols}

The impossibility of achieving 
unconditionally secure two-party cryptographic protocols has led to the definition and use of simpler hardware or software primitives to form the basis for the desired functionalities. 
A one out of $n$ oblivious transfer (OT) is a two party protocol where one party (Alice) has input $n$ messages $(x_1,\cdots,x_n)$ and the other party (Bob) inputs a number $c\in\{1,\cdots,n\}$ and receives the 
message $x_c$ with the following guarantees: Bob learns nothing about the other messages $x_i|i\neq c$ and Alice is ``oblivious'' (does not know) which message Bob obtained (i.e. does not know the value $c$) \cite{NP99}. A hardware token that implements a non-interactive OT is called one-time-memory (OTM) \cite{GKR08}. We will utilise a one out of $n$ OTM, where again Alice stores in the OTM $n$ strings $(x_1,\cdots,x_n)$, Bob specifies a value $c$ and the OTM reveals $x_c$ to Bob and then self-destructs: i.e. the remaining strings $x_i|i\neq c$ are lost forever.

The first paper that studied secure two party quantum computation is \cite{DNS10} and we follow their notations and conventions. We have two parties $A,B$ with registers $\mathcal{A},\mathcal{B}$ and an extra register $\mathcal{R}$ with $\dim \mathcal{R}=(\dim\mathcal{A}+\dim\mathcal{B})$. The input state is denoted $\rho_{in}\in D(\mathcal{A}\otimes\mathcal{B}\otimes\mathcal{R})$, where $D(\mathcal{A})$ is the set of all possible quantum states in register $\mathcal{A}$. We also denote with $L(\mathcal{A})$ the set of linear mappings from $\mathcal{A}$ to itself, and  the superoperator $\phi:L(\mathcal{A})\rightarrow L(\mathcal{B})$ that is completely positive and trace preserving, is called a \emph{quantum operation}. 
We denote $\mathbb{I}_{\mathcal{A}}$ 
the identity operator 
in register $\mathcal{A}$. The ideal output is then given by $\rho_{out}=(U\otimes\mathbb{I}_{\mathcal{R}})\cdot \rho_{in}$, where for simplicity we write $U\cdot \rho$ instead of $U\rho U^\dagger$. For two states $\rho_0,\rho_1$ we denote the trace norm distance $\Delta(\rho_0,\rho_1):=\frac12 \lVert \rho_0-\rho_1\rVert$. If $\Delta(\rho_0,\rho_1)\leq\epsilon$ then any process applied on $\rho_0$ behaves as for $\rho_1$ except with probability at most $\epsilon$.

\begin{definition}[taken from \cite{DNS10}] A $n$-step two party strategy with oracle calls is denoted $\Pi^O=(A,B,O,n)$:
\begin{enumerate}
\item input spaces $\mathcal{A}_0,\mathcal{B}_0$ and memory spaces $\mathcal{A}_1,\cdots,\mathcal{A}_n$ and $\mathcal{B}_1,\cdots,\mathcal{B}_n$
\item $n$-tuple of quantum operations $(L_1^A,\cdots,L_n^A)$ and $(L_1^B,\cdots,L_n^B)$ such that $L_i^A: L(\mathcal{A}_{i-1})\rightarrow L(\mathcal{A}_i)$ and similarly for $L_i^B$.
\item $n$-tuple of oracle operations $(\mathcal{O}_1,\cdots,\mathcal{O}_n)$ where each oracle is a global operation for that step, $\mathcal{O}_i:L(\mathcal{A}_i\otimes\mathcal{B}_i)\rightarrow L(\mathcal{A}_i\otimes\mathcal{B}_i)$
\end{enumerate}
\end{definition}
The trivial oracle is a communication oracle, that transfers some quantum register from one party to another. Protocols that have only such oracles are called bare model. Other oracles include calls to other cryptographic primitives. The quantum state  in each step of the protocol is given by:
\EQ{\rho_1(\rho_{in})&:=&(\mathcal{O}_1\otimes\mathbb{I})(L_1^A\otimes L_1^B\otimes\mathbb{I})(\rho_{in})\nonumber\\
\rho_{i+1}(\rho_{in})&:=&(\mathcal{O}_{i+1}\otimes\mathbb{I})(L_{i+1}^A\otimes L_{i+1}^B\otimes\mathbb{I})(\rho_i(\rho_{in}))
}


The security definitions are based on the ideal functionality of two party quantum computation 
(2PQC) that takes a joint input $\rho_{in}\in \mathcal{A}_0\otimes\mathcal{B}_0$, obtains the state $U\cdot\rho_{in}$ and returns to each party their corresponding quantum registers. A protocol $\Pi^O_U$ implements the protocol securely, if no possible adversary in any step of the protocol, can 
distinguish whether they interact with the real protocol or with a simulator that has only access to the ideal functionality. When a party is malicious we add the notation ``$\sim$'', e.g. $\tilde A$.

\begin{definition}[Simulator]
$\mathcal{S}(\tilde{A})=\langle (\mathcal{S}_1,\cdots,\mathcal{S}_n),q \rangle$ is a simulatore for adversary $\tilde{A}$ in $\Pi^O_U$ if it consists of:
\begin{enumerate}
\item operations where $\mathcal{S}_i:L(\mathcal{A}_0)\rightarrow L(\tilde{\mathcal{A}_i})$,
\item  sequence of bits $q\in\{0,1\}^n$ determining if the simulator calls the ideal functionality at step $i$ ($q_i=1$ calls the ideal functionality).
\end{enumerate}
\end{definition}
Similarly for the other adversaries.

Given input $\rho_{in}$ the simulated view for step $i$ is defined as:
\EQ{\nu_i(\tilde A,\rho_{in}):=\Tr_{\mathcal{B}_0}\left((\mathcal{T}_i\otimes\mathbb{I})(U^{q_i}\otimes\mathbb{I})\cdot \rho_{in}\right)
}
and similarly for the other party. 

\begin{definition}[Privacy] \label{def:private}We say that the protocol is $\delta$-private if for all adversaries and for all steps $i$: 
\EQ{\Delta(\nu_i(\tilde{A},\rho_{in}),\Tr_{\mathcal{B}_i}(\tilde{\rho}_i(\tilde{A},\rho_{in})))\leq\delta}
where $\tilde{\rho}_i(\tilde{A},\rho_{in})$ the state of the real protocol with corrupted party $\tilde{A}$, at step $i$. 
\end{definition}
In classical cryptography a type of adversary commonly considered, is the ``honest-but-curious''. This adversary follows the protocol but also keeps records of their actions and attempts to learn from those more than what they should. This type of weak adversary has been proven very useful in many protocols, since it typically constitutes the first step in constructing protocols secure against more powerful (even fully malicious) adversaries. 

Since quantum states cannot be copied, one cannot have a direct analogue of honest-but-curious. Instead, we have the notion of \emph{specious} adversaries \cite{DNS10}, where they can deviate as they wish, but in every step, if requested, should be able to reproduce the honest global state by acting only on their subsystems. More formally:

\begin{definition}[Specious]
An adversary $\tilde{A}$ is $\epsilon$-specious if there exists a sequence of operations $(\mathcal{T}_1,\cdots,\mathcal{T}_n)$, where $\mathcal{T}_i: L(\tilde{\mathcal{A}}_i)\rightarrow L(\mathcal{A}_i)$ such that:
\EQ{\label{eq:specious}
\Delta\left((\mathcal{T}_i\otimes\mathbb{I})(\tilde\rho_i(\tilde A,\rho_{in})),\rho_i(\rho_{in})\right)\leq\epsilon
}
\end{definition}
Note, that this (standard) definition of specious adversary, allows for different interpretations that lead to stronger and weaker versions of the adversary. 
These subtleties are explained in Appendix \ref{app:specious}. Here we stress that we take the weaker notion that takes Eq. (\ref{eq:specious}) in the stricter sense, while in \cite{DNS10} the authors implicitly used a stronger version. This difference led to our paper evading certain impossibility results mentioned in \cite{DNS10}. A detailed discussion of this is given in Appendix \ref{app:specious}.




A 2PQC protocol needs to be (by definition) private. Moreover, it may (or may not) have two extra properties, \emph{verification} and \emph{fairness}. Verification means that each party, when they receive their output, not only are sure that nothing leaked about their input, but also they can know whether the outcome they received is correct or corrupted. Fairness is the extra property that no party should be able to obtain their output and \emph{after} that cause an abort (or a corruption) for the other party. Even in classical 2PC, it is impossible to have fairness against fully malicious adversaries (without any extra assumption). We consider two-party protocols $\Pi^0=(C,S,O,n)$ where $C$ denotes the client and $S$ the server.

\section{Secure two-party quantum computation (QYao) \label{sec:2PQC}}

The first extension of VUBQC that we will explore, is to use it in order to construct a 2PQC protocol similar to classical Yao protocol of \cite{Yao86}, that we refer to as QYao, with the sender-garbler being the client, and the receiver-evaluator being the server. As in Yao protocol, we assume that the client, when adversarial, is specious, however, we make no such assumption for the server, that is assumed to be fully malicious. In \cite{Yao86} the server needs to use OT in order to insert their input. We, instead, use a scheme to insert the quantum input that requires no such functionality. The reason this is possible, is  because we make the assumption that the client is specious and in specific situations a specious adversary is weaker than honest-but-curious (see Appendix \ref{app:specious-classical}). For the specific case of classical input/output, one can modify our protocol to make it secure against classical honest-but-curious adversary, by replacing the input injection subprotocol (see below) with OT.

Our QYao protocol provides a one-time \emph{verifiable} 2PQC similar to the classical setting as recently shown for the original Yao protocol in \cite{GGB10}. The speciousness of the client restricts any possible deviations on their side, while a malicious server would be detected through the hidden traps of the VUBQC protocol. The blindness property of VUBQC also guarantees that the server learns nothing before the client is certain that there was no deviation and returns the suitable keys for the decryption of the output.

The major differences that QYao protocol has in comparison with regular VUBQC \cite{KW15} is that the server needs to provide (part of) the input and at the end keeps (part of) the output. There are multiple ways to modify the VUBQC protocol, we use a direct approach at the cost of having two rounds of quantum communication during input preparation and two rounds of quantum communication during output read-out. However, there is \emph{no} quantum communication during the evaluation stage. Note that if we restrict the protocol to classical input/output, we would avoid \emph{all} quantum communication apart from the initial and offline sending of pre-rotated qubits from the client to the server (but we would then need OT, as in the classical Yao, for the server to insert their input).

\subsection{Server's input injection}

The qubits composing the DT(G), (see Figure \ref{figure1}), are prepared by the client, this is crucial in order to ensure that the server is blind about the positions of the traps. However, the server should somehow insert their input in the DT(G). We present here, for simplicity, the case that there is a single qubit input, but it generalises trivially.
To do so, the server encrypts their input using secret keys $(m_{z,i},m_{x,i})$ and sends the state $X^{m_{x,i}}Z^{m_{z,i}}\ket{\Phi_S}$ to the client, to insert it randomly in the corresponding base-location.  The state $\ket{\Phi_S}$ is assumed to be pure as we have included its purification in the hands of the server. The client, encrypts further the server's input by applying an extra random $Z(\theta'_i)$ and an $X^{x'_i}$ correction to obtain the state:

\EQ{X^{x'_i}Z(\theta'_i)X^{m_{x,i}}Z^{m_{z,i}}\ket{\Phi_S}
\label{input1}} 
This extra encryption is needed in order to hide the future actual measurement of the input qubit (to be performed by the server) and ensure that no information about 
the position of the traps is leaked. 
Trap hiding also requires the client to return two extra qubits for the input insertion (see Figure \ref{figure2}) that will belong to trap graphs (called the white-graph and black-graph). The white qubit is a trap is prepared in a state $\ket{\theta_k}$ while the black qubit is dummy and is prepared in state $\ket{d_j}$. The three qubits will be randomly permuted by the client so that the server does not know which qubit is which. A similar procedure (with no communication from the server) is applied for client's input qubits, as well as for all the qubits corresponding to the gate computation (see Figure \ref{figure2}).

\begin{figure}
\includegraphics[width=0.8\columnwidth]{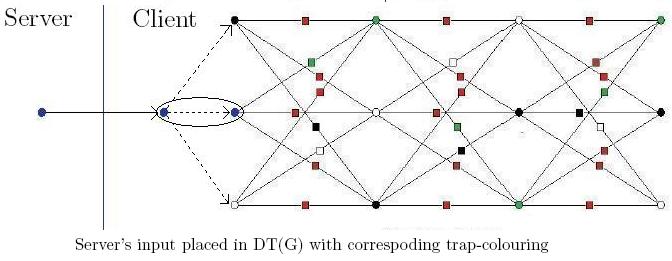}

\caption{Server gives their input (blue) and client chooses (randomly) where in the input base-location to place the input. The random choice is highlighted. The trap-colouring is filled correspondingly, after the random choice is made.}
\label{figure2}
\end{figure}

Finally, after the server has received all the qubits, announces the secret keys $(m_{x,i},m_{z,i})$ for each input $i$ to the client, so that the client can update the encryption for these qubits and have $(x_i,\theta_i):=\left(x'_i+m_{x,i},(-1)^{m_{x,i}}\theta'_i+\pi m_{z,i}\right)$. With the updated encryption, the client computes the suitable measurement angles $\delta_i$. It is worth pointing out that the key releasing step from server to client could be avoided, by using classical OT to compute the measurement angles as a function of the secret parameters of the server $\delta_i(m_{x,i},m_{z,i})$ for the first two layers (that have dependency on $m_{x,i},m_{z,i}$).  While this could be necessary for future work, to construct protocols dealing with malicious client, it is not necessary for our case where the client is considered to be specious.

\begin{algorithm}[H]
\caption{Server's Input Injection}
\label{prot:prover-input}

\begin{flushleft}
\textbf{Setting:}
\vspace{-7pt}
\end{flushleft}

\noindent $\bullet$ The server has input $\ket{\Phi_S}$ that corresponds to specific positions $I_S$ of the input layer of the base-graph. For each $i$ qubit, server chooses pair of secret bits $(m_{x,i},m_{z,i})$.

\begin{flushleft}
\textbf{Instructions:}
\vspace{-7pt}
\end{flushleft}

\begin{enumerate}
\item The server sends to the client the states $X^{m_{x,i}}Z^{m_{z,i}}\ket{\Phi_S}$ for all $i\in I_S$.

\item The client prepares all the states of the DT(G) as in Protocol \ref{prot:KW15} (see \cite{KW15}) apart from the computation qubits of the server's input.

\item The client chooses at random $x'_i,\theta_i'$ for each of the server's input and obtains the states: $X^{x'_i}Z(\theta'_i)X^{m_{x,i}}Z^{m_{z,i}}\ket{\Phi_S}$. 

\item The client mixes the computation qubit of the server's input, with a dummy and a trap qubit to return the three qubits of server's input base-locations.

\item The server, after receiving the qubits, returns to the client the secret bits $(m_{x,i},m_{z,i})$ for all $i$ of their input.

\item The client computes $x_i:=x'_i+m_{x,i}$ and $\theta_i:=(-1)^{m_{x,i}}\theta'_i+\pi m_{z,i}$ and uses these $(x_i,\theta_i)$ for computing the measurement angles as in Protocol \ref{prot:KW15}.

\end{enumerate}

\begin{flushleft}
\textbf{Outcome:}
\vspace{-7pt}
\end{flushleft}

\noindent$\bullet$ The server receives a DT(G), where the quantum input is the joint input of the client and server and the related measurement angles instructions perform the desired unitary, if the client is honest. 

\end{algorithm}

\subsection{Server's output extraction}

In the regular VUBQC protocol, the server returns all the output qubits to the client. The client measures the final layer's traps to check for any deviation and \emph{then} obtains the output of the computation by decrypting the output computation qubits using their secret keys. In the 2PQC, part of the output (of known base-locations) should remain in the hands of the server. This, however, would not allow the client to check for the related traps (that could have effects on other output qubits). Similar to the input injection, the solution is obtained via an extra layer of encryption by server followed by a delayed key releasing. The server will encrypt using two classical bits $(k_x,k_z)$: $E_{k}(\ket{\psi})=X^{k_x}Z^{k_z}\ket{\psi}$, all the output qubits that correspond to server's output base-locations and \emph{then} return all the output qubits to the client. Due to the encryption, the client obtains no information about the output of the server while now has 
access to all the final trap qubits for the verification. The client returns to the server, only the computation qubits corresponding to server's output base-locations, while keeping the traps and all the qubits from other base-locations. The server reveals their keys (note that the client, being specious, has not kept the computation output qubits, because it would be impossible to reconstruct the ideal state by acting only on their systems). The client checks all the traps, and if all of them are correct, reveals the final computation output keys to the server ($\theta_i,r_i$). The server undoes the padding to reveal their output.  

\begin{algorithm}[H]
\caption{Server's Output Extraction}
\label{prot:prover-output}

\begin{flushleft}
\textbf{Setting:}
\vspace{-7pt}
\end{flushleft}

\noindent $\bullet$ The server has the state $\Tr_C(\ket{\psi(f)})$ at the corresponding positions $O_S$ of their output base-locations of the DT(G). For each $i$ qubit, server chooses pair of secret bits $(m_{x,i},m_{z,i})$. 

\begin{flushleft}
\textbf{Instructions:}
\vspace{-7pt}
\end{flushleft}

\begin{enumerate}
\item The server sends to the client the states $X^{m_{x,i}}Z^{m_{z,i}}\Tr_C(\ket{\psi(f)})$ for all $i\in O_S$.

\item The client keeps the qubits that correspond to traps and dummies in the final layer of server's output, while returns the computation qubits of server's output.

\item The server returns to the client $(m_{x,i},m_{z,i})$.

\item The client checks the traps and if correct returns to the server the final layer paddings for their output locations $(\theta_i,r_i,r_{j<i})$ and $(m_{x,i},m_{z,i})$ (that the server sent earlier).

\item The server undoes the final layer paddings and obtains the output.

\end{enumerate}

\begin{flushleft}
\textbf{Outcome:}
\vspace{-7pt}
\end{flushleft}

\noindent$\bullet$ The server obtains the computation qubits of their output base-locations unpadded $\Tr_C(U(\ket{\Psi}_C,\ket{\Phi}_S))$. 

\end{algorithm}

\subsection{The QYao protocol}

We combine Protocol \ref{prot:prover-input} and \ref{prot:prover-output} with the core VUBQC, Protocol \ref{prot:KW15}, to obtain a secure two-party protocol to compute unitaries given as Protocol \ref{prot:Qyao}. 

\begin{algorithm}[H]
\caption{Secure 2PQC - QYao}
\label{prot:Qyao}

\begin{flushleft}
\textbf{Input:}
\vspace{-7pt}
\end{flushleft}
\begin{itemize}

\item The client and server know the unitary operator $U$ that they wish to compute. The client has a description of $U$ in MBQC using resource $\ket{G}$, and maps it to the DT(G) (see Protocol \ref{prot:KW15} and \cite{KW15}). For each qubit the client knows the angles $\phi_i$ and dependencies. The base-locations of the inputs and outputs of both parties are public. The inputs of client and server are denoted correspondingly as $\ket{\Psi}_C$ and $\ket{\Phi}_S$.
\end{itemize}

\begin{flushleft}
\textbf{Output:}
\vspace{-7pt}
\end{flushleft}

\begin{itemize}
\item The client receives the subset $C_o$ of the output qubits of the final quantum state $U(\ket{\Psi}_C,\ket{\Phi}_S)$. 

\item The server receives the subset $S_o$ of the output qubits of the final quantum state $U(\ket{\Psi}_C,\ket{\Phi}_S)$. 
\end{itemize}

\begin{flushleft}
\textbf{The protocol}
\vspace{-7pt}
\end{flushleft}

\begin{enumerate}
\item Client sends all qubits of the DT(G) (choosing the random parameters $r_i,\theta_i,d_i$ and trap-colouring), where for base-locations that corresponds to the server's input, the input injection Protocol \ref{prot:prover-input} is used.

\item Server and client follow the verification protocol \ref{prot:KW15} 
until the step that the output is returned. 

\item Client and server interact according to Protocol \ref{prot:prover-output} so that the server extracts their own output.

\end{enumerate}

\end{algorithm}

\begin{Theorem}[Correctness]\label{thm:correctness} If both client and server follow the steps of Protocol \ref{prot:Qyao} then the output is correct and the computation accepted
\end{Theorem}

\begin{proof}
If client and server follow Protocol \ref{prot:Qyao}, after Step 1, where the server injected their input, we are in exactly the same situation as in Protocol \ref{prot:KW15}, with (overall) input $\rho_{in}=\ket{\Psi}_C\otimes\ket{\Phi}_S$.

During Step 2, client and server run exactly Protocol \ref{prot:KW15} and correctness follows from \cite{KW15}. The positions of dummies result in having isolated traps measured in correct basis and thus there is never an abort. From the remaining qubits, one copy of the dotted base-graph $G$, where the measurement pattern $\mathbb{M}_{\textrm{Comp}}$ is applied, results to the state $E_{k^S,k^C}(U\cdot \rho_{in})$. This is the honest global final state, encrypted with keys of both the client (secret parameters) and the server (the padding $(m_{x,i},m_{z,i})$ from the output extraction protocol, which applies to server's output registers only).

During Step 3, the client keeps the registers of their output, while returns the registers of the server's output. The protocol finishes with the server returning their keys $(m_{x,i},m_{z,i})$ to the client (to check for traps), 
while the client returns the keys (secret parameters) involved with 
the final decryption of the server's output registers.
\end{proof}

Due to the simple composition of input injection and output extraction the verification property of our QYao is directly inherited from the VUBQC. We prove this first, 
before presenting and proving the main privacy property of the QYao protocol in the next section (which exploits the verifiability).

\begin{definition}\label{verification} 
We define a 2PQC protocol to be \textbf{$\epsilon$-verifiable for the client}, if for any (potentially malicious) server, the probability of obtaining a corrupt output \emph{and} not abort is bounded by $\epsilon$. The output of the real protocol with malicious server $\tilde S$ is $\tilde\rho(\tilde S,\rho_{in})$ and we have:

\EQ{\label{eq:verification}\Delta(\tilde\rho(\tilde S,\rho_{in}),\rho_{ideal}(\rho'_{in}))&\leq&\epsilon}
where 

\AR{\rho_{ideal}(\rho_{in})&:=&p_{ok}(\mathbb{I}_{\H_C}\otimes \mathcal{C}_{\H_S})\cdot U\cdot(\rho_{in})+(1-p_{ok})(\ket{fail}\bra{fail})
}
and $\mathcal{C}_{\H_S}$ is the deviation that acts on the server's systems after they receive their outcome (a CP-map, but can be purified if we include ancilla). Also $\rho'_{in}=(\mathbb{I}_{\H_C}\otimes D_{\H_S})\cdot\rho_{in}$ is an initial state compatible with the client's input, with 
$D_{\H_S}$ the 
deviation on the input by the server. 
\end{definition}
We should note, that the server can always choose \emph{any} input from their side, and the security of the protocol is defined with respect to this ``deviated'' input. Moreover, since the deviation $\mathcal{C}_{\H_S}$ is performed at the final step of the protocol, we also have that the global state (before that deviation, i.e. at step $n-1$) obeys:

\EQ{\label{eq:verification2}\Delta(\tilde\rho^{n-1}(\tilde S,\rho_{in}),\rho^{n-1}_{ideal}(\rho'_{in}))&\leq&\epsilon}
where 
\AR{
\rho^{n-1}_{ideal}(\rho_{in})&:=&p_{ok} U\cdot(\rho_{in})+(1-p_{ok})(\ket{fail}\bra{fail})
}


\begin{Theorem}[$\epsilon$-verification for client]\label{thm:verification}
Protocol \ref{prot:Qyao} is $\epsilon$-verifiable for the client, where $\epsilon=\left(\frac{8}{9}\right)^d$ and $d=\lceil\frac\delta{2(2c+1)}\rceil$,  $c$ is the maximum degree of the base graph and $\delta$ is the number of errors tolerated on the base graph $G$.
\end{Theorem}
\begin{proof}
In order to prove the verifiability for the client, 
we assume that the client is honest, while the server is malicious.

During Step 1 of Protocol \ref{prot:Qyao}, the server sends their input and also gives the keys of the one-time-padding encryption  $(m_{x,i},m_{z,i})$ from the input injection phase. It follows, that any deviation on these affects only the computation qubits of the server's input base-location, in other worlds resulting to a state: $\rho_{in}'=(\mathbb{I}_{\H_{C}}\otimes D_{\H_{S}})\cdot\rho_{in}$

During Step 2, the protocol proceeds exactly as the VUBQC, Protocol \ref{prot:KW15}, with only difference that the  qubits with base-location of the server's output, have an extra encryption with keys known to the server. This means that the client delays the measurement of the traps in those base-locations, until the next step, where they receive from the server these keys. Note, however, that the client can already check all the past traps and the ones corresponding to base-location of the clients output. 

In Step 3, the server returns the keys $(m_{x,i},m_{z,i})$ for all their output base-locations qubit, and also receives the computation qubits of those base-locations encrypted with both the client's and server's keys. Any deviation by the server at this stage, either returns different (wrong) keys to the client, or acts only on the server's output. Returning wrong keys increases the chance of abort, but in any case, it is equivalent with this deviation happening \emph{before} the return of the output to the server. The only remaining, extra, deviation is a possible deviation acting on the computation output qubits of the server, i.e.  a deviation of the form $(\mathbb{I}_{\H_C}\otimes \mathcal{C}_{\H_S})$.

It follows, that this protocol has the same verification properties with Protocol \ref{prot:KW15}, with the modified input and (server's-part) of output given by Eq. (\ref{eq:verification}). From \cite{KW15} we have that Protocol \ref{prot:KW15} is $\epsilon$-verifiable with $\epsilon=\left(\frac{8}{9}\right)^d$ and $d=\lceil\frac\delta{2(2c+1)}\rceil$, and this completes the proof.
\end{proof}

It is worth mentioning, that this verification property essentially restricts the server to behave similarly with a specious adversary, with the extra ability to abort the protocol.

\section{Proof of Privacy of the QYao protocol \label{sec:simulators}}

Recall that we defined a protocol to be secure if no possible adversary in any step of the protocol, can distinguish whether they interact with the real protocol or with a simulator that has access only to the ideal functionality (see Definition \ref{def:private}).

\begin{Theorem}\label{thm:privacy}
The QYao protocol, Protocol \ref{prot:Qyao}, that is $\epsilon_1$-verifiable for the client, is $O(\sqrt{\epsilon_2})$-private against an $\epsilon_2$-specious client and $\epsilon_1$-private against a malicious server.
\end{Theorem}

To prove the theorem we need to introduce simulators for each step of the protocol and each possible adversary. Below we define those simulators and prove that they are as close to the real protocol as requested from Theorem \ref{thm:privacy}. Since in our setting the two parties have different roles and maliciousness we consider the simulators for each party separately.

\subsection{Client's simulators}

The client is $\epsilon$-specious and this means that for each step $i$ there exist a map $\mathcal{T}_i: L(\H_{\tilde V_i})\rightarrow L(\H_{V_i})$, that 

\EQ{
\Delta\left((\mathcal{T}_i\otimes\mathbb{I})(\tilde\rho_i(\tilde V,\rho_{in})),\rho_i(\rho_{in})\right)\leq\epsilon
}
Following the proof of the ``Rushing Lemma'' of \cite{DNS10}, we obtain a similar lemma for each step of the protocol, where it shows that there is no extra information stored in the ancilla's of the specious adversary:

\begin{Lemma}[No-extra Information]\label{rushing1}
Let $\Pi_U=(A,B,n)$ be a correct protocol for two party evaluation of $U$. Let $\tilde A$ be any $\epsilon$-specious
adversary. Then there exists an isometry $T_i:\tilde A_i\rightarrow A_i\otimes \hat{A}$ and a (fixed) mixed state $\hat{\rho_i}\in D(\hat{A_i})$ such that for all joint input states $\rho_{in}$,

\EQ{\label{eq:rushing}
\Delta\left((T_i\otimes\mathbb{I})(\tilde\rho_i(\tilde A,\rho_{in})),\hat{\rho_i}\otimes\rho_i(\rho_{in})\right)\leq 12\sqrt{2\epsilon}
}
where $\rho_i(\rho_{in})$ is the state in the honest run and $\tilde\rho_i(\tilde A,\rho_{in})$ is the real state (with the specious adversary $\tilde A$).
\end{Lemma}

\begin{proof} This proof follows closely the proof of the ``rushing lemma'' of \cite{DNS10}, where one can find further details. For simplicity we assume pure $\rho_{in}$ (where it holds in general by convexity). Consider two different states $\ket{\psi_1},\ket{\psi_2}$ in $\mathcal{A}_0\otimes\mathcal{B}_0\otimes\mathcal{R}$, and we extend the space $\mathcal{R'}=\mathcal{R}\otimes\mathcal{R}_2$ with $\mathcal{R}_2=\textrm{span}\{\ket{0},\ket{1}\}$. We define the state $\ket{\psi}=1/\sqrt{2}(\ket{\psi_1}\ket{1}+\ket{\psi_2}\ket{2})$. Due to the speciousness of $\tilde{A}$ and using Ulmann's theorem, there is an isometry $T_i:\tilde{\mathcal{A}}_i\rightarrow \mathcal{A}_i\otimes \hat{\mathcal{A}}$ and a state $\tilde{\rho}\in D(\hat{\mathcal{A}})$ such that

\EQ{\Delta\left((T_i\otimes\mathbb{I})\cdot \tilde{\rho_i}(\tilde{A},\psi),\tilde{\rho}\otimes\rho_i(\psi)\right)\leq 2\sqrt{2\epsilon}}
The state $\tilde{\rho}$ in general is not independent of the input $\psi$ so there are few more steps required. By noting that projection and partial trace are trace non-increase maps (by projecting on $\ket{1}$ subspace and tracing out the $\mathcal{R}'$ subspace), we obtain:

\EQ{\Delta\left((T_i\otimes\mathbb{I})\cdot \tilde{\rho_i}(\tilde{A},\psi_1),\tilde{\rho}\otimes\rho_i(\psi_1)\right)\leq 4\sqrt{2\epsilon}}
and similarly for $\psi_2$. We repeat the same using $\psi_1,\psi_3$ and with initial state $\ket{\psi'}=1/\sqrt{2}(\ket{\psi_1}\ket{1}+\ket{\psi_3}\ket{2})$  and obtain:

\EQ{\Delta\left((T_i\otimes\mathbb{I})\cdot \tilde{\rho_i}(\tilde{A},\psi_1),\tilde{\rho'}\otimes\rho_i(\psi_1)\right)\leq 4\sqrt{2\epsilon}}
By the triangular inequality we get $\Delta(\tilde{\rho},\tilde{\rho'})\leq 8\sqrt{2\epsilon}$. This means that for any state $\ket{\chi}$, there exists a state $\hat{\rho}\in\hat{\mathcal{A}}$ with $\Delta(\tilde{\rho},\hat{\rho})\leq 8\sqrt{2\epsilon}$ such that:

\EQ{\Delta\left((T_i\otimes\mathbb{I})\cdot \tilde{\rho_i}(\tilde{A},\chi),\hat{\rho}\otimes\rho_i(\chi)\right)\leq 4\sqrt{2\epsilon}}
and the Lemma follows by using once more the triangle inequality. 
\end{proof}

Before introducing the simulators we give some intuitive arguments. The client, during the input-injection step, receives a fully one time padded quantum state, while in later steps (up until the last) has no legitimate quantum state (all the quantum states are in the server's side).
Therefore, using the above Lemma, we can already see that the client has no information about the actual quantum state during any step before the last. This is illustrated by the fact that the simulators for these steps can run the honest protocol with a random (but fixed) input and the client's view is the same (they cannot distinguish a run with the correct or different input in any step before the last). The simulator for the final step (after receiving the quantum states from the server) is more subtle, as it necessarily involves a call to the ideal functionality.

\begin{proof} of Theorem \ref{thm:privacy}.

\hskip 02cm

\noindent\textbf{Simulator upon receiving server's encoded input}: There is a full one time pad on server's input, thus a simulator can use a random fixed state $\varphi^*$ instead of the real state $\Tr_{C}(\rho_{in})$ and the client cannot distinguish them. 

\hskip 02cm

\noindent\textbf{Simulator after returning server's input and before receiving output}: The simulator runs the honest protocol using instead of $\rho_{in}$ a fixed random input $\varphi^*$. The honest state with input $\chi$ at this stage is:

\EQ{\label{eq:10}
\rho_i(\chi)&&:= \\&& \nonumber\left(\ket{k_i}\bra{k_i}\otimes\ket{b_i}\bra{b_i}\right)_{\H_C}\otimes \left( E_{k_i}(\sigma_i(\chi))\otimes \ket{d_i}\bra{d_i}\otimes\ket{+_{\theta_{t_i}}}\bra{+_{\theta_{t_i}}} \right)_{\H_S}
}
where $\sigma_i(\chi)$ is the evolution of the input state, when the gates corresponding to the $i$th step have been performed (here we collectively call the secret parameters used to encrypt the state in each step as $k_i$). Later, we will be more specific about the different secret parameters. Also $d_i,\theta_{t_i}$ are the dummy and trap qubits of the $i$th layer (we can include the qubits of future layers as well with no difference in the remaining argument). It is easy to see that the reduced state on the client's side $\H_C$ is independent from the input state $\chi$. The simulator is then constructed:

\begin{enumerate}
\item The simulator runs the protocol with fixed but random input $\varphi^*$, to obtain $\rho_i(\varphi^*)$.
\item The simulator obtains the fixed state $\hat{\rho_i}$ and the isometry $T_i$, that both are independent of the input.
\item The simulated view is then defined to be:

\EQ{
\nu_i(\tilde C,\rho_{in}):=\Tr_{
\H_{S_i}} \left( (T_i^\dagger\otimes\mathbb{I})(\hat{\rho_i}\otimes\rho_i(\varphi^*))\right)
}

\end{enumerate}

\noindent We can easily see that 

\EQ{\nu_i(\tilde C,\rho_{in})=\Tr_{
\H_{S_i}} \left( (T_i^\dagger\otimes\mathbb{I})(\hat{\rho_i}\otimes\rho_i(\rho_{in}))\right)
=\Tr_{
\H_{S_i}} \left( (T_i^\dagger\otimes\mathbb{I})(\hat{\rho_i}\otimes\rho_i(\varphi^*))\right)\nonumber
}
as $T_i^\dagger$ does not acts on $\H_{S_i}$ and thus commutes with the partial trace, while tracing-out $\H_{S_i}$ leaves the reduced state independent from the input (see Eq. (\ref{eq:10})). Since isometries leave invariant the trace-distance, while partial trace is non-increasing, we use Eq. (\ref{eq:rushing}) and we have

\EQ{
\Delta\left(\Tr_{\H_{S_i}}\left(\tilde\rho_i(\tilde C,\rho_{in})\right),\nu_i(\tilde C,\rho_{in})\right)\leq 12\sqrt{2\epsilon}
}

\noindent\textbf{Simulator after receiving client's output}: First we need to introduce some notation and conventions: 

\begin{itemize}

\item By $k^S,k^C$ we define the padding of the server's and client's output respectively. In more detail these keys are functions of the trap-colouring (trap positions) of the last layer, the rotations $\theta$, the secret parameters $r$'s and measurement results $b$'s of the previous two layers. We collectively denote the rest, i.e. all the other secret keys, as $k^R$. 

\item We also assume that the classical measurement outcomes $b$'s are in the Hilbert space of the client, while, being classical and public, could be copied to the server's Hilbert space too (adding this copy does not affect the client's simulators and complicates the notation). Moreover, classical results corresponding to trap measurements are denoted $b_t$, while other are simply $b$. The trap measurements outcomes (in the honest run) are $b_t=r_t$ where $r_t$ is one of the secret parameters included in $k^R$.

\item The final layer (unmeasured) qubits are separated to: (i) output qubits (those are of the server and the client, and it includes the purification of their inputs), (ii) dummy qubits collectively denoted $d^S,d^C$ for server/client unmeasured dummies and (iii) trap qubits denoted as $\theta_t^S,\theta_t^C$.  The dummies and trap qubits are in tensor product to the rest (in the honest run).

\item The Hilbert spaces containing the output qubits (and the related purification of honest run) are denoted $\H_{S_o},\H_{C_o}$, while the remaining qubits are $\H_S,\H_C$.

\item For brevity we denote $[b]:=\ket{b}\bra{b}$ and similarly for the other states.

\end{itemize}
With the above notations, the honest run of the protocol at this step, with some input $\chi$ is:

\EQ{
\rho_i(\chi)&=&\left(E_{k_i^C,k_i^S}(U(\chi))\right)_{\H_{C_o}\otimes\H_{S_o}}\nonumber\\ 
& &\otimes\left([k^S]\otimes[k^C]\otimes[k^R]\otimes[b]\otimes[b_t]\otimes[d^C]\otimes[\theta_t^C]\right)_{\H_C}\nonumber\\ 
& &\otimes \left([d^S]\otimes[\theta_t^S]\right)_{\H_S}
}
We define:

\EQ{
\varrho(\chi,k^C,k^S)&:=& \left(E_{k_i^C,k_i^S}(U(\chi))\right)_{\H_{C_o}\otimes\H_{S_o}}\nonumber\\
\sigma&:=&\left([k^S]\otimes[k^C]\otimes[k^R]\otimes[b]\otimes[b_t]\otimes[d^C]\otimes[\theta_t^C]\right)_{\H_C}\otimes \left([d^S]\otimes[\theta_t^S]\right)_{\H_S}\nonumber\\
\rho_i(\chi)&=&\varrho(\chi,k^C,k^S)\otimes\sigma
}
It is worth pointing-out that $\sigma$ is independent from the input $\chi$. Now the simulator is constructed in the following steps:
\begin{enumerate}
\item The simulator obtains from the client the choice of secret keys $k^C,k^S,k^R$ (that can be viewed as part of the client's input).

\item The simulator sets $b_t=r_t$. For the other measurement outcomes $b$, the simulator chooses randomly a bit value. To ensure that the $b$'s obtained are indistinguishable from the real honest protocol, the simulator can run the full protocol with some random input $\varphi$. Since in the honest run, the values of $b$'s do not depend on (are not correlated to) the input, the outcomes that the simulator returns would be indistinguishable from those of a run with the correct input. Finally, speciousness ensures that this state is also close to the real protocol.

\item The simulator, using the previous steps, constructs the state $\sigma$.

\item The simulator uses the state $\hat{\rho_i}$ from Lemma \ref{rushing1} for the $i$th step, which is a fixed state (independent of the input).

\item The simulator calls the ideal functionality and receives the state $\Tr_{\H_{S_o}}\left(\rho_i(\rho_{in})\right)$.

\item The simulator uses the keys and the state received from the ideal functionality and constructs $\Tr_{\H_{S_o}}\left(\varrho(\rho_{in},k^C,k^S)\right)$. 

\item The simulator obtains the operator $\mathcal{T}_i$ from the definition of specious, and constructs the isometry $T_i$ acting from $\H_{\tilde C}\rightarrow\H_{C}\otimes\H_{\hat{C}}$.

\item The simulated view is then:

\EQ{
\nu(\tilde C,\rho_{in}):=\Tr_{\H_S}\left((T^\dagger_i\otimes\mathbb{I})\left(\hat{\rho_i}\otimes\Tr_{\H_{S_o}}\left(\varrho(\rho_{in},k^C,k^S)\right)\otimes\sigma\right)\right)
}

\end{enumerate}
It is clear, that in all this construction, $\rho_{in}$ appears only through $\Tr_{\H_{S_o}}\left(\rho_i(\rho_{in})\right)$, which is the ideal output of the client. Now we prove that this simulated view is $\delta$-close to the view of the client in the real protocol.

\noindent We start from Eq. (\ref{eq:rushing}) and we obtain:

\EQ{
\Delta\left(\Tr_{\H_S,\H_{S_o}}\left(\tilde\rho_i(\tilde C,\rho_{in})\right),\Tr_{\H_S,\H_{S_o}}\left(\left(T^\dagger_i\otimes\mathbb{I}\right)\left(\hat{\rho_i}\otimes\rho_i(\rho_{in})\right)\right)\right)\leq 12\sqrt{2\epsilon}
\nonumber}
Now, since $T_i^\dagger$ does not act on either of the server's Hilbert spaces, it commutes with the partial trace. We can then see that:

\EQ{
\nu(\tilde C,\rho_{in})=\Tr_{\H_S,\H_{S_o}}\left(\left(T^\dagger_i\otimes\mathbb{I}\right)\left(\hat{\rho_i}\otimes\rho_i(\rho_{in})\right)\right)
}
and this concludes the first part of the proof as we have a    $(12\sqrt{2\epsilon})$-private protocol.

\end{proof}

\subsection{Server's simulators}

Our QYao protocol is secure against a fully malicious server that can deviate in any possible way and can also cause an abort. We will use the fact that QYao protocol is $\epsilon$-verifiable for the client, i.e. Eq. (\ref{eq:verification}) and Eq. (\ref{eq:verification2}) hold. In many classical protocols, to prove security against malicious adversaries one has to restrict their actions to essentially honest-but-curious adversaries. Here the verification property plays such a role, as the condition obtained can be used in a similar way as the speciousness of the client.

We can see that, before the client sends the secret parameters, the server is totally blind, i.e. the server's reduced state at all times is the totally mixed state (for all qubits in the DT(G), including their own input qubits, after they are injected). This is highlighted by the fact that one can run the full protocol (before releasing the keys), without choosing the (client's) input (see simulator below). After receiving the keys from the client, to simulate the server's view we need a call to the ideal functionality, and at this point to use Eq. (\ref{eq:verification}). 

\vskip 0.2cm
\noindent \emph{Continuation of the proof of Theorem \ref{thm:privacy}.}

\noindent\textbf{Simulator before receiving keys}: 
The simulator, instead of sending qubits in one of these states $\{\ket{+_\theta},\ket{0},\ket{1}\}$, sends one side of an EPR pair $\ket{\psi}=\frac{1}{\sqrt{2}}(\ket{01}+\ket{10})$ for each qubit to the server. Then chooses an angle $\delta_i$ at random for each of the measurement angles. The simulator can measure (if they wish) their qubits (of the EPR pairs) in suitable way and angles that inserts any input they wish and performs any computation. This can be done after \emph{all} measurements of the server have taken place (see 
\cite{DFPR13}). It follows that the server cannot obtain any information about the client's input. 


\noindent\textbf{Simulator after receiving keys}: The simulator is constructed using the following steps:

\begin{enumerate}

\item Simulator prepares multiple Bell states $\ket{\psi}=\frac{1}{\sqrt{2}}(\ket{01}+\ket{10})$. Sends the one qubit of each pair to the server and instructs them to entangle the qubits as they would in the normal protocol. 

\item Simulator for each qubit chooses randomly an angle $\delta_i$ and instructs the server to measure in this angle.

\item Simulator obtains from the server's (malicious) strategy the parameter $p_{ok}$ 
(see Eq. (\ref{eq:verification})).

\item With probability $(1-p_{ok})$ the simulator returns abort. Otherwise performs the remaining steps.

\item Simulator calls the ideal functionality and obtains the state $\Tr_{\H_C}(U(\rho'_{in}))$, where $\rho'_{in}=(\mathbb{I}_{\H_C}\otimes D_{\H_S})\cdot \rho_{in}$ is the deviated input that the corrupted server 
inputs.



\item 
The simulator encrypts the outcome using $k^S$, and sends the output qubits of the server back in the state: $E_{k^S}\left(
\Tr_{\H_C}(U(\rho'_{in}))\right)$

\item  The simulator returns the keys $k^S$ to the server.

\end{enumerate}

\noindent The last step of the protocol is that the server decrypts their output, which we denote as $D_{k^S}(\cdot)$, where $D_{k^S}(E_{k^S}(\cdot))=\mathbb{I}$. It follows, that the server's view of the real protocol after the key exchange is:

\EQ{\Tr_{\H_C}\left(E_{k^S}\left(\tilde\rho^{n-1}(\tilde S,\rho'_{in})\right)\right)
}
From Eq. (\ref{eq:verification2}) we obtain

\EQ{\Delta\left(\Tr_{\H_C}\left(E_{k^S}\left(\tilde\rho^{n-1}(\tilde S,\rho_{in})\right)\right),\Tr_{\H_C}\left(E_{k^S}(\rho^{n-1}_{ideal}(\rho'_{in}))\right)\right)\leq\epsilon
}
as partial trace is distance non-increasing operation. Using the definition of $\rho^{n-1}_{ideal}(\rho'_{in})$ from Eq. (\ref{eq:verification2}) we see that

\EQ{\Tr_{\H_C}\left(E_{k^S}(\rho^{n-1}_{ideal}(\rho'_{in}))\right)&=&\Tr_{\H_C}\left(E_{k^S}\left(p_{ok} U\cdot(\rho'_{in})+(1-p_{ok})(\ket{fail}\bra{fail})\right)\right)\nonumber\\
\Tr_{\H_C}\left(E_{k^S}(\rho^{n-1}_{ideal}(\rho'_{in}))\right)&=&\nu(\tilde P,\rho_{in})
}
We have now proved that the simulated view is $\epsilon$-close to the real view of the server, just after the key exchange, and thus proving that the protocol is $\epsilon$-private against malicious server $\tilde S$. \qed

\section{Non-interactive QYao 
\label{sec:non-interactive}}

Following the classical approach of \cite{GKR08} we exploit our QYao protocol to construct one-time program. To do so we simply need to remove the classical online communication of the QYao protocol using the classical hardware primitive of secure ``one-time memory'' (OTM), which is essentially a non-interactive oblivious transfer. 
The obtained one-time quantum programs can be executed only once, where the input can be chosen at any time. Classically the challenge in lifting the Yao protocol for two-party computation to  one-time program, was the issue of malicious adversary. However, our QYao protocol is already secure against a malicious evaluator without any extra primitive or added overhead.

Recall that the interaction in our QYao is required for two reasons. First, from the server's perspective this is done to obtain the measurement angles $\delta_i$ that perform the correct computation, while these could not be computed offline as they depend on the measurement outcomes $b_{j<i}$ of certain qubits measured before $i$ but after the preparation stage  (and on secret parameters $(\theta_i,r_i,r_{j<i},T)$ and computation angle $\phi_i$ that are known from start). Second, from the client's perspective, the results of measurements are needed to provide with a ``proof'' that the output is correct, by testing for deviation from the trap outcomes.

\noindent\textbf{Removing the interaction:} An obvious solution for the first issue raised, is to have the client compute $\delta_i$ for all combinations of previous outcomes $b_{j<i}$, and then store in an OTM the different values, while the server chooses and learns the entry of the OTM corresponding to the outcomes obtained. This solution, at first, appears to suffer from an efficiency issue as one may think that for each qubit the client needs to compute $\delta_i$ for \emph{all} combination of past outcomes which grows exponentially (as would the size of the OTM used). However, a closer look in the dependencies of corrections in MBQC for the typical graphs used, 
shows that the measurement angle of qubit $i$ depends on at most a constant number of qubits, being within the two previous ``layers'' (past neighbours or past neighbours of past neighbours). This is evident from the flow construction \cite{DK2006,BKMP2007}, that guarantees that corrections can be done, and the explicit form of dependencies involves only neighbours and next-to neighbours in the graph $G$.  A flow is defined by a function ($f:O^c\rightarrow I^c$) from measured qubits to non-input qubits and a partial order $(\preceq)$ over the vertices of the graph such that $\forall i:i\preceq f(i)$ and $\forall j\in N_G(i): f(i)\preceq j$, where $N_G(i)$ denotes the neighbours of $i$ in $G$. Each qubit $i$ is $X$-dependent on $f^{-1}(i)$ and $Z$-dependent on all qubits $j$ such that $i\in N_G(j)$.

\begin{definition}[Past of qubit $i$] We define $P_i=Z_i\cup X_i$ to be the set of qubits $j$ that have $X$ or $Z$ dependency on $i$.
\end{definition}
\begin{definition}[Influence-past of qubit $i$]
We define influence-past $c_i$, of qubit $i$ to be an assignment of an outcome $b_j\in\{0,1\}$ for all qubits $j\in P_i$.  
\end{definition}
For each influence-past $c_i$, there exists a unique value of $\delta_i(c_i)$. While $c_i$ provides all the necessary dependencies for the client to compute $\delta_i$, this is not known to the server. This could be problematic, as the server is expected to open the OTM using the past outcomes with a labelling consistent with their own knowledge. The \emph{true} dependencies depend on the actual flow of the computation, which is hidden from the server, as the positions of dummies (that break the DT(G) to the three graphs) are not known to the server. From here on we restrict attention to DT(G). This is resolved by defining: 
\begin{definition}[Extended-past of qubit $i$]
We define the \emph{extended past $EP_i$ of qubit $i$} to be the set of all qubits that for some trap-colouring are in the past of $i$.
\end{definition}
Similarly we define \emph{extended-influence-past} (EIP) of qubit $i$ (it is clear that $\delta_i$ has trivial dependency on all outcomes of EIP that are not in the actual influence-past). The extended past has also finite cardinality, as it is evident that the only qubits that can affect a qubit with base-location $i$ are those that belong to (the finite) base-locations that contain the qubits that have true dependency (i.e. neighbours or next-to neighbours base-locations).

\noindent\textbf{Ensuring verification:} The interaction is also important from the clients perspective, in order to verify that the server did not deviate. Again, naively one could suggest that the server can return all the measurement outcomes at the end. However, this suffers from a potential attack (not present in normal VUBQC). The underlying reason is that there is nothing to force the server to commit that the $b_i$'s that will return are the same as the ones used for selecting measuring angles. For example, the server measures qubit $1$ at $\delta_1$, and obtains $b_1$ (which will be returned at the end) but uses $b_1'=b_1\oplus 1$ when opening the OTM corresponding to qubit $j$ where $1\in EP_j$. If qubit $j$ is computation qubit, this leads to measuring at a wrong angle, while if qubit $j$ is trap qubit, it has no (real) dependence on $b_1$ and thus the angle recovered is the correct one, and the server never causes abort while corrupts the computation. To avoid this problem, we introduce a method that uses some flag bit-string $l_i$.

The client provides for each qubit, the angle $\delta_i$ and a flag bit-string $l_i$ of length $m$. This flag bit-string is a (one-time-padded) flag, of whether a trap failed in the (extended) past of qubit $i$ or not. 
At the final step of the protocol, the server returns the final layer qubits and \emph{all} the flags $l_i$. 
Here we will denote $b_{j;i}$ to mean the outcome of qubit $j$ measurement that the client uses for the computation of the measurement angle for qubit $i$, i.e. the value of $b_j$ selected by the server when opening the $i$th OTM.

The client chooses uniformly at random, for each qubit $i$ an $m$-bit string $l^0_i$, to denote that none of its past qubits ($j\in P_i$) was a failed trap. The string $l^0_i$ will be called ``accept'' flag, while all other strings $l_i\neq l^0_i$ are the ``reject'' flags. The OTM's are prepared according to Protocol \ref{prot:OTM}.

\begin{algorithm}[H]
\caption{Preparation of OTM's for non-interactive QYao} 
 \label{prot:OTM}

\noindent\textbf{Description of OTM's:} We will use an (at most) one out of $K$ OTM, where $K=2^{\max_i|EP_i|}$. Note, that the first layer does not need an OTM, since (i) we give the angle directly (does not depend on anything) and (ii) there is no trap in the previous layer and thus no need to give a flag. Moreover, for the last layer, while the qubits are not measured (and thus there is no corresponding angle $\delta$) we will use an OTM to obtain the last flags (that correspond to testing traps in the previous layers).

\begin{itemize}
\item \textbf{Labels} of the cells of the OTM at $i$: Each qubit is labelled according to the outcomes of the extended past qubits $b_{j;i}|j\in EP_i$, in other words according to the extended influence past $c_i$. 

\item\textbf{Content} of the cells of the OTM at $i$: In each cell of the OTM we have a pair of numbers $(\delta_i,l_i)$, (that can be viewed as a single string where the first three bits give $\delta$ and the remaining the flag). The $\delta_i$ is the unique correct $\delta$ for the particular outcomes of past qubits $b_j|j\in P_i$.

\EQ{\delta_i(c_i)=(-1)^{s_i^{X_i}}\phi_i+\theta_i+\pi(r_i\oplus s_i^{Z_i})}
The flag $l_i$ depends on the outcomes $b_j$ of the extended past $j\in EP_i$.
In particular:

\begin{enumerate}
\item If for all traps $t$ in the extended past of $i$, i.e. $t\in EP_i$, 
$b_{t;i}=r_t$, we return $l^0_i$, i.e. accept flag.
\item Otherwise, we return a random string $l_i\neq l^0_i$, i.e. reject flag.
\end{enumerate}
\end{itemize}
\end{algorithm}
According to this protocol, the flag that the server obtains when opening the OTM at qubit $i$, is reject if and only if in the extended past of this qubit there is a failed trap outcome. We can see that returning a wrong outcome for some measurement $b_i$, implies opening an entry of the OTM with wrong flag only if the qubit is a trap one (as in regular VUBQC), while it still returns the accept flag if it is a dummy or computation qubit. Here we should note, that even if the client knows that one flag is a reject flag (i.e. has one particular $l_i\neq l_i^0$), the probability of guessing the correct flag is only $\epsilon_f:=(2^m-1)^{-1}$. This, intuitively, will force the client to return the flag obtained from the OTM (or abort with high probability), provided that $m$ is chosen suitably. We can now give the non-interactive QYao Protocol \ref{prot:NI-QYao}. 

\begin{algorithm}[H]
\caption{Non-interactive QYao} 
 \label{prot:NI-QYao}
 \vskip 0.2 cm
\textbf{Assumptions} \\ Client and server want to jointly compute a unitary as in Protocol \ref{prot:Qyao}. The client has $N$ OTM's that are $1$-out-of-$K$, i.e. one OTM per qubit, with sufficient entries to store a pair of $(\delta_i,l_i)$ measurement angle and flag bit-string (of length at least $m\geq\log (\frac1\epsilon+1)$), for each extended influence past of the qubit.




\textbf{Protocol}
\begin{enumerate}
\item Client and server interact according to Protocol \ref{prot:prover-input} to obtain the server's input locations. Client also sends the qubits of DT(G) aftr choosing secret parameters ($r_i,\theta_i,d_i$ and trap-colouring).

\item The client, for each qubit and each extended influence past, computes $\delta_i(c_i)$. Then prepares one OTM per qubit as described in Protocol \ref{prot:OTM}. 

\item Server performs the measurements according to the first layer of angles received $\delta_i$ (directly as in Protocol \ref{prot:prover-input}). Then open the next layer OTM's using the outcomes $b_i$ of their measurements. Uses the new measurement angle revealed $\delta_j$, while records the flag bit-string $l_j$. Iterates until the last layer OTM is opened (and the second last layer is measured). The final layer OTM's return only a flag.

\item Server and client interact according to Protocol \ref{prot:prover-output} so that the server obtains their output. The only difference is that the client, before returning the keys, in order to accept checks the flags (instead of checking the trap outcomes of measured qubits) and the final layer traps.

\end{enumerate}
\end{algorithm}


\begin{Theorem}\label{thm:verification0}
Protocol \ref{prot:NI-QYao} is $\epsilon$-verifiable for the client, with $\epsilon$ same as the VUBQC protocol that is used.
\end{Theorem}

\begin{proof}
This proof consists of three stages. First, similarly to the proof of theorem \ref{thm:verification}, we show that the verifiability for the client of the non-interactive QYao, reduces to the verifiability of the same VUBQC protocol with modifications for the server's input and output. The second stage is to observe that the optimal strategy for an adversarial server is to return the flags from the opened OTM's. This makes the verifiability property for this protocol identical with an interactive protocol with the only modification that the server can return different values for the measurement outcome $b_i$ depending on which future qubit the client needs the outcome ($b_{i;j}$). The final stage, is to show that this modified (interactive) verification protocol, is $\epsilon$-verifiable with same $\epsilon$ as Protocol \ref{prot:KW15}.

\noindent\textbf{Stage 1:} Following the proof of theorem \ref{thm:verification}, we can see that the non-interactive QYao protocol \ref{prot:NI-QYao} is $\epsilon$-verifiable for the client if the corresponding non-interactive VUBQC protocol that is used during step 3 of Protocol \ref{prot:NI-QYao} is $\epsilon$-verifiable, with deviated input $\rho'_{in}$ and a final deviation $(\mathbb{I}_{\H_C}\otimes \mathcal{C}_{\H_S})$ on the server's output (computation) qubits.

\noindent\textbf{Stage 2:} 
If the server attempts to guess a flag (given the extra knowledge of one bit-string $l_i\neq l_i^0$), they cause an abort with high probability ($1-\epsilon_f$, where $\epsilon_f:=(2^m-1)^{-1}$). However, since by assumption $m\geq\log (\frac1\epsilon+1)$, the probability of abort is so high that makes the protocol trivially verifiable, as it is $\epsilon$-close to the ideal state Eq. (\ref{eq:verification}) with $(1-p_{ok})=(1-\epsilon_f)$.

It follows that the adversary (trying to maximising their cheating chances) should return the flags found in the OTMs. This is equivalent with an interactive VUBQC protocol, that (i) the server for each qubit $i$ returns multiple values of the measurement outcome $b_{i;j}$, one for each $j$ in the extended future qubit of $i$ and (ii) the client uses those outcomes to compute the $\delta_i$'s and (iii) the client aborts only when they receive at least one trap outcome wrong $b_{t;j}\neq r_t$ for any $j$.

\noindent\textbf{Stage 3:} We should now show the $\epsilon$-verifiability of the modified interactive VUBQC protocol described above. The proof follows the same steps of the proof of verifiability of Protocol \ref{prot:KW15} in \cite{KW15}. The first steps exploit the blindness to reduce the possible attacks (of adversarial server) to convex combination of Pauli attacks. Then it is noted, that since the computation is encoded in an error-correcting code (that corrects up to $\delta/2$ errors), there is a minimum number ($d=\lceil\frac\delta{2(2c+1)}\rceil$) of independent base-locations that need to be corrupted to cause an error. 

For the proof of verifiability, as in \cite{fk,KW15}, we make the assumption that if the minimum number of attacks that could corrupt the computation occurs, then the computation \emph{is} corrupted. This is clearly not true, but is sufficient to provide a bound on the probability of corrupt \emph{and} accept (which is the ``bad'' case, that the server manages to cheat). Then, given this minimum number of attacks, the probability of abort is computed and found to be greater than $1-\epsilon$, and thus the protocol is verifiable. 

In our case, there is the difference that for each measured qubit $b_i$ of the original protocol, we have multiple qubits $b_{i;j}$. Once again, we make the assumption of minimum corrupted qubits, we need at least $d=\lceil\frac\delta{2(2c+1)}\rceil$ independent base-locations, and we can allow for each base-location to corrupt a single $b_{i;j}$ for one specific $j$. However, this does not change the probability of hitting a trap, as it suffices to give wrong value $b_{t,j}$ for one $j$ to cause an abort. We then obtain again, that the probability of abort (in the minimum corruptable attack) is at least $(1-\epsilon)$, (for $\epsilon=\left(\frac89\right)^d$) and the protocol is indeed $\epsilon$-verifiable.
\end{proof}

\begin{Theorem}
Protocol \ref{prot:NI-QYao} that is $\epsilon_1$-verifiable for the client, is $O(\sqrt{\epsilon_2})$-private against an $\epsilon_2$-specious client and $\epsilon_1$-private against a malicious server.
\end{Theorem}

\begin{proof} We need simulators for the client 
only during input injection and output extraction (since the client does not participate in the evaluation phase, in the non-interactive protocol). During these steps, the simulators are identical with those in Protocol \ref{prot:Qyao}.

The server's simulators until before the output extraction, can be identical with the ones of Protocol \ref{prot:Qyao}, where the interaction is replaced with the preparation of OTM's. It is important to note that all OTM's can be constructed with \emph{no} information about the client's input and thus from the simulator. Finally, the simulator for the final step of Protocol \ref{prot:Qyao}, needs to only use the property that the Protocol is $\epsilon$-verifiable for the client. We proved in Theorem \ref{thm:verification0} that this is the case for Protocol \ref{prot:NI-QYao} and thus we can use the same simulator.  
\end{proof}


\acknowledgments{We would like to thank Fr\'{e}d\'{e}ric Dupuis for discussions on the definition of specious adversaries. Funding from EPSRC grants  EP/N003829/1 and EP/M013243/1 is acknowledged. }

\appendix
\section{Measurement-based quantum computation \label{app:mbqc}}

An MBQC pattern is fully characterised by the graph (representing the entangled state), default measurement angles $\phi_i$, an order of measurements (determined by the flow) and a set of corrections that determine the actual measurement angle of each qubit (modify the default measurement angles with respect to previous measurement outcomes).

For the brickwork state, or any subset of the square lattice state, the flow $f(i)$ takes a qubit to the same row next column. E.g. $i=(k,l)$, then $f(i)=(k+1,l)$. Given a flow $f(\cdot)$, and measurement outcomes $s_j$, the actual (corrected) measurement angle is given:

\EQ{
\phi'_i(\phi_i,s_{j<i}):=\phi_i(-1)^{s_{f^{-1}(i)}}+\pi \sum_{j|f(j)\in N_G(i)\setminus f(i)}s_j
}
where $N_G(i)$ denotes the set of neighbours (in the graph) of vertex $i$. It is easy to see that the corrected angle for qubit $i$ has an $X$-correction (i.e. a $(-1)$ factor before $\phi_i$) from the one qubit $f^{-1}(i)$ and $Z$-corrections (i.e. $\pi$ addition) from some neighbours of neighbours, which for the brickwork state is at most two (but in any case are constant in number).

We present here diagrams taken from \cite{bfk} showing how to translate a universal set of gates to the MBQC measurement patterns using the brickwork graphs.

\begin{figure}[H]
\includegraphics[width=1\columnwidth]{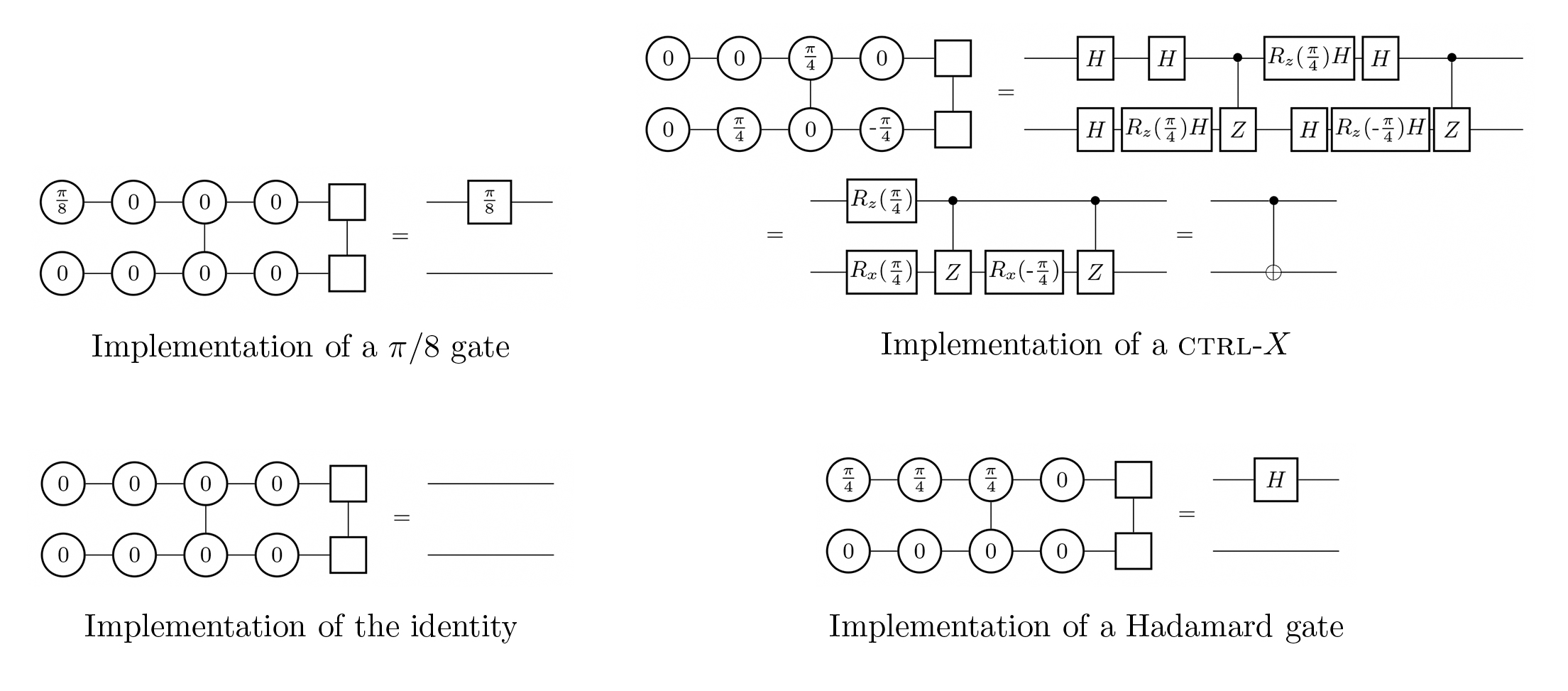}

\label{mbqc-paterns}
\end{figure}

Finally, we note that in UBQC, the measurement angle that the client instructs the server to measure is given: $\delta_i=\phi_i'(\phi_i,b_{j<i})+\theta_i+\pi r_i=C(i,\phi_i,\theta_i,r_i,\mathbf{s})$

\section{On specious adversaries definitions \label{app:specious}}

The definition of specious, as given by Eq. (\ref{eq:specious}), includes $\rho_i(\rho_{in})$ which is the honest state in step $i$, for \emph{any} input state $\rho_{in}$. There are two subtle issues with what this means, (i) in relation with what are the possible input states $\rho_{in}$, and (ii) in relation to whether the random secret parameters are considered ``input of the computation or not. If we take the most restrictive case, it leads to a very weak adversary. In \cite{DNS10}, there are some impossibility results, that stem from taking a (slightly) stronger form of this adversary. Here we see separately these two issues and finally discuss the differences with \cite{DNS10}.

\subsection{Restricting to classical inputs}\label{app:specious-classical}

The first important point is that we observe  
that a specious adversary, under certain conditions, could be weaker than an honest-but-curious classical adversary. A specious adversary is allowed to do actions/operations, that for any (allowed) input, can be ``undone'' by actions on his side if there is an audit.

Now we consider a (trivial) example, that specious is weaker than classical honest-but-curious. Assume, that as first step of a protocol, a party (that is specious adversary) $\tilde{A}$ receives a one-time-padded quantum input $E_k(\ket{\psi}_S)$ of the other party. As second step, the $\tilde{A}$ returns the padded quantum state back.  If the input is considered to be a general (unknown) quantum state, a specious adversary cannot do the following action before returning the system $S$:  

\EQ{(\wedge X)_{SA} E_k(\ket{\psi}_S)\otimes\ket{0}_A
} 
There is no map that $\tilde{A}$ can apply to their private system/ancilla $A$ alone, and obtain the correct state $\rho_i(\rho_{in})$, because for a general $E_k(\ket{\psi}_S)$ the resulting joint state is entangled. 

However, imagine that we are actually considering is a classical protocol which means that the input is in computational basis, i.e. either $\ket{0}$ or $\ket{1}$. In that case, the $\wedge X$ simply copies the ciphertext $E_k(\ket{\psi}_S)$, which is exactly the action that an honest-but-curious classical adversary \emph{can} do. Note that in this specific case, the resulting state is no longer an entangled state and thus could be recovered by acting only on system $S$. Nonetheless, the definition of specious requires to recover the correct state at each step for \emph{any} possible input. 

It is exactly because of this property (that specious under certain conditions is weaker than classical honest-but-curious) that we can avoid using OT for inserting the input of the server (unlike the Yao protocol). In generalisations where we will have stronger adversaries, we will once again need to have OT for the input insertion.

\subsection{About the secret random parameters}

The second important subtlety of the specious definition, is related with the secret (random) parameters that the adversary can choose.
Specious adversaries should be able to recover the global state $\rho_i(\rho_{in})$ in any step. However this definition may be somewhat ambiguous. In general, the quantum state in any step is also function of secret random parameters of the two parties $k_A,k_B$, i.e. we have $\rho_i(\rho_{in},k_A,k_B)$. We can (and generally do) treat the secret keys as part of the input of the two parties (i.e. part of $\rho_{in}$), but in this case the specious adversary is essentially requested to reconstruct precisely the state $\rho_i(\rho_{in},k_A,k_B)$ at step $i$.

However, we can imagine an adversary $\tilde A$ that could reproduce the state $\rho_i(\rho_{in},k'_A,k_B)$ instead, i.e. reproduce a state that would be correct for step $i$, if the secret parameter was $k'_A$ instead of the $k_A$ that was given at the start. In this case, the adversary $\tilde A$ is not specious with the standard definition that we use. On the other hand, intuitively, since $k_A$ is a secret parameter (not known by anyone but $\tilde A$ until this step), it should not matter what is this value, and reproducing $\rho_i(\rho_{in},k'_A,k_B)$ should be sufficient for a version of quantum honest-but-curious.

We will define \emph{strong specious} adversary, to be the adversary that they are required to have CP maps $\mathcal{T}_i$ acting only on their subsystem such that they can reproduce the state $\rho_i(\rho_{in},k_A,k_B)$ for at least one secret key $k_A$ (not determined in advance, and thus of their choice).

We will again give a simple example to make the distinction of these two flavour of specious adversaries. Imagine a protocol that a (strong or weak) specious adversary $\tilde{A}$ is supposed to return an unknown state $\ket{\psi}$ padded with their key $E_{k_A}(\ket{\psi})$. A strong specious adversary can cheat by keeping the state $\ket{\psi}$ and returning instead the one side of an EPR pair (and keep the other side). Then if an audit occurs, the adversary can use their side of the EPR pair and teleport the state $\ket{\psi}$ back, where the resulting state at the honest side is $E_{k_m}(\ket{\psi})$, with $k_m$ being (randomly) determined by the outcome of the Bell measurement that teleports the state. A weak specious adversary, on the other hand, cannot do this. The state that is supposed to be returned needs to be padded with the key $k_A$ that is fixed from the start of the protocol, while $k_m$ is randomly determined during the audit (by the Bell measurement outcome).

\subsection{Regarding secure SWAP no-go theorem}

It is not difficult to see, that a protocol that is secure against the standard (weak) specious adversary can be made secure against the strong specious adversary by modifying the protocol to request that in every step a new random (secret) parameter appears, a commitment to its value is made. This means that requesting an adversary to be weak specious is practically equivalent with strong specious adversary where commitments are allowed. 

Here it is worth mentioning, that in \cite{DNS10}, while not explicitly stated, had their protocol secure against the stronger version of specious adversaries (unlike our protocol). This difference in the definition of specious adversaries (implicit assumption of the stronger adversaries) also resolves the apparent contradiction of our result (no secure primitive needed) with their no-go theorem (a secure SWAP is needed for 2PQC even for specious adversaries). As is mentioned in \cite{DNS10}, their no-go would not hold if commitments were possible, which is exactly what our weaker definition essentially permits.

\bibliographystyle{unsrt}
\bibliography{report}

\begin{thebibliography}{10}

\bibitem{bfk}
Anne Broadbent, Joseph Fitzsimons, and Elham Kashefi.
\newblock Universal blind quantum computation.
\newblock In {\em Proceedings of the 50th Annual Symposium on Foundations of
  Computer Science}, FOCS '09, pages 517 -- 526. IEEE Computer Society, 2009.

\bibitem{fk}
Joseph~F. Fitzsimons and Elham Kashefi.
\newblock Unconditionally verifiable blind computation, 2012.
\newblock Eprint:\href{http://arxiv.org/abs/1203.5217}{arXiv:1203.5217}.

\bibitem{KW15}
Elham Kashefi and Petros Wallden.
\newblock Optimised resource construction for verifiable quantum computation.
\newblock {\em Journal of Physics A: Mathematical and Theoretical; preprint
  arXiv:1510.07408}, 2017.

\bibitem{Yao86}
Andrew Yao.
\newblock How to generate and exchange secrets.
\newblock In {\em Foundations of Computer Science, 1986., 27th Annual Symposium
  on}, pages 162--167. IEEE, 1986.

\bibitem{DNS10}
Fr{\'e}d{\'e}ric Dupuis, Jesper~Buus Nielsen, and Louis Salvail.
\newblock Secure two-party quantum evaluation of unitaries against specious
  adversaries.
\newblock In {\em Advances in Cryptology--CRYPTO 2010}, pages 685--706.
  Springer, 2010.

\bibitem{GKR08}
Shafi Goldwasser, Yael~Tauman Kalai, and Guy~N Rothblum.
\newblock One-time programs.
\newblock In {\em Advances in Cryptology--CRYPTO 2008}, pages 39--56. Springer,
  2008.

\bibitem{BGS13}
Anne Broadbent, Gus Gutoski, and Douglas Stebila.
\newblock Quantum one-time programs.
\newblock In Ran Canetti and JuanA. Garay, editors, {\em Advances in Cryptology
  – CRYPTO 2013}, volume 8043 of {\em Lecture Notes in Computer Science},
  pages 344--360. Springer Berlin Heidelberg, 2013.

\bibitem{Childs}
A.~Childs.
\newblock Secure assisted quantum computation.
\newblock {\em Quant. Inf. Compt.}, 5(6):456, 2005.

\bibitem{AS06}
P.~Arrighi and L.~Salvail.
\newblock Blind quantum computation.
\newblock {\em International Journal of Quantum Information}, 4:883--898, 2006.

\bibitem{abe}
Dorit Aharonov, Michael Ben-Or, and Elad Eban.
\newblock Interactive proofs for quantum computations.
\newblock In {\em Proceedings of Innovations in Computer Science 2010},
  ICS2010, pages 453--, 2010.

\bibitem{ruv2}
Ben~W. Reichardt, Reichardt~Falk Unger, and Umesh Vazirani.
\newblock Classical command of quantum systems.
\newblock {\em Nature}, 496:456--460, 2013.

\bibitem{onewaycomputer}
Robert Raussendorf and Hans~J. Briegel.
\newblock A one-way quantum computer.
\newblock {\em Phys. Rev. Lett.}, 86:5188--5191, May 2001.

\bibitem{NFB14}
Naomi~H Nickerson, Joseph~F Fitzsimons, and Simon~C Benjamin.
\newblock Freely scalable quantum technologies using cells of 5-to-50 qubits
  with very lossy and noisy photonic links.
\newblock {\em Physical Review X}, 4(4):041041, 2014.

\bibitem{BKBFZW11}
S.~Barz, E.~Kashefi, A.~Broadbent, J.~F. Fitzsimons, A.~Zeilinger, and
  P.~Walther.
\newblock Demonstration of blind quantum computing.
\newblock {\em Science}, 335(6066):303, 2012.

\bibitem{KKD14}
Theodoros Kapourniotis, Elham Kashefi, and Animesh Datta.
\newblock {Blindness and Verification of Quantum Computation with One Pure
  Qubit}.
\newblock In {\em 9th Conference on the Theory of Quantum Computation,
  Communication and Cryptography (TQC 2014)}.

\bibitem{gkw2015}
Alexandru Gheorghiu, Elham Kashefi, and Petros Wallden.
\newblock Robustness and device independence of verifiable blind quantum
  computing.
\newblock {\em New Journal of Physics}, 17(8):083040, 2015.

\bibitem{kdk2015}
Theodoros Kapourniotis, Vedran Dunjko, and Elham Kashefi.
\newblock On optimising quantum communication in verifiable quantum computing,
  2015.
\newblock Eprint:\href{http://arxiv.org/abs/1506.06943}{arXiv:1506.06943}, to
  appear in the Proceedings of the 15th Asian Quantum Information Science
  Conference (AQISC 2015).

\bibitem{DFPR13}
Vedran Dunjko, Joseph~F. Fitzsimons, Christopher Portmann, and Renato Renner.
\newblock Composable security of delegated quantum computation.
\newblock In {\em Advances in Cryptology}. 2014.

\bibitem{Mckague13}
McKague.
\newblock Interactive proofs for bqp via self-tested graph states.
\newblock arXiv:1309.5675, 2013.

\bibitem{B15}
Anne Broadbent.
\newblock How to verify a quantum computation.
\newblock {\em arXiv preprint arXiv:1509.09180}, 2015.

\bibitem{GWK2017}
Alexandru Gheorghiu, Petros Wallden, and Elham Kashefi.
\newblock Rigidity of quantum steering and one-sided device-independent
  verifiable quantum computation.
\newblock {\em New Journal of Physics; preprint arXiv:1512.07401}, 2017.

\bibitem{DNS12}
Fr{\'e}d{\'e}ric Dupuis, Jesper~Buus Nielsen, and Louis Salvail.
\newblock Actively secure two-party evaluation of any quantum operation.
\newblock In {\em Advances in Cryptology--CRYPTO 2012}, pages 794--811.
  Springer, 2012.

\bibitem{DK16}
Vedran Dunjko and Elham Kashefi.
\newblock Blind quantum computing with two almost identical states.
\newblock {\em arXiv preprint arXiv:1604.01586}, 2016.

\bibitem{MR11}
Ueli Maurer and Renato Renner.
\newblock Abstract cryptography.
\newblock In {\em In Innovations in Computer Science}. Citeseer, 2011.

\bibitem{Kent99}
Adrian Kent.
\newblock Unconditionally secure bit commitment.
\newblock {\em Physical Review Letters}, 83(7):1447, 1999.

\bibitem{BBE92}
Charles Bennett, Gilles Brassard, and Artur Ekert.
\newblock Quantum cryptography.
\newblock In {\em Progress in Atomic physics Neutrinos and Gravitation,
  proceedings of the XXVIIth Rencontre de Moriond Series: Moriond Workshops,
  held January}, page 371, 1992.

\bibitem{BS16}
Anne Broadbent and Christian Schaffner.
\newblock Quantum cryptography beyond quantum key distribution.
\newblock {\em Designs, Codes and Cryptography}, 78(1):351--382, 2016.

\bibitem{GF16}
Gorjan Alagic and Bill Fefferman.
\newblock On quantum obfuscation.
\newblock {\em arXiv preprint arXiv:1602.01771}, 2016.

\bibitem{Broadbent15}
Anne Broadbent.
\newblock How to verify a quantum computation.
\newblock {\em arXiv preprint arXiv:1509.09180}, 2015.

\bibitem{HT15}
Masahito Hayashi and Tomoyuki Morimae.
\newblock Verifiable measurement-only blind quantum computing with stabilizer
  testing.
\newblock {\em arXiv preprint arXiv:1505.07535}, 2015.

\bibitem{DSS16}
Yfke Dulek, Christian Schaffner, and Florian Speelman.
\newblock Quantum homomorphic encryption for polynomial-sized circuits.
\newblock {\em arXiv preprint arXiv:1603.09717}, 2016.

\bibitem{KP16}
Elham Kashefi and Anna Pappa.
\newblock Blind multiparty quantum computing.
\newblock {\em In preparation}, 2016.

\bibitem{childs2005unified}
Andrew~M Childs, Debbie~W Leung, and Michael~A Nielsen.
\newblock Unified derivations of measurement-based schemes for quantum
  computation.
\newblock {\em Physical Review A}, 71(3):032318, 2005.

\bibitem{mbqc}
Vincent Danos, Elham Kashefi, and Prakash Panangaden.
\newblock The measurement calculus.
\newblock {\em J. ACM}, 54(2), April 2007.

\bibitem{hein2004multiparty}
Marc Hein, Jens Eisert, and Hans~J Briegel.
\newblock Multiparty entanglement in graph states.
\newblock {\em Physical Review A}, 69(6):062311, 2004.

\bibitem{NP99}
Moni Naor and Benny Pinkas.
\newblock Oblivious transfer and polynomial evaluation.
\newblock In {\em Proceedings of the thirty-first annual ACM symposium on
  Theory of computing}, pages 245--254. ACM, 1999.

\bibitem{GGB10}
Rosario Gennaro, Craig Gentry, and Bryan Parno.
\newblock Non-interactive verifiable computing: Outsourcing computation to
  untrusted workers.
\newblock In {\em Advances in Cryptology--CRYPTO 2010}, pages 465--482.
  Springer, 2010.

\bibitem{DK2006}
Vincent Danos and Elham Kashefi.
\newblock Determinism in the one-way model.
\newblock {\em Phys. Rev. A}, 74:052310, Nov 2006.

\bibitem{BKMP2007}
Daniel~E Browne, Elham Kashefi, Mehdi Mhalla, and Simon Perdrix.
\newblock Generalized flow and determinism in measurement-based quantum
  computation.
\newblock {\em New Journal of Physics}, 9(8):250, 2007.

\end{thebibliography}

\end{document}